\documentclass[a4paper, onecolumn, 10pt, unpublished]{quantumarticle}
\pdfoutput=1
\usepackage{mathtools,amssymb}
\usepackage{amsthm}
\usepackage{nicematrix}
\usepackage[dvipsnames]{xcolor}
\usepackage{caption}
\usepackage{subcaption}
\usepackage{tikz}
\usetikzlibrary{cd}
\usepackage{stmaryrd}
\usepackage{braket}
\usepackage{bbm}
\usepackage{dsfont}
\usepackage{thm-restate}
\usepackage{quiver}

\newenvironment{proofsketch}{%
  \begin{proof}%
}{%
  \end{proof}%
}

\usepackage{graphicx}
\graphicspath{{./figures/}}

\usepackage[citestyle=numeric-comp, sorting=none]{biblatex}
\usepackage{hyperref}
\hypersetup{
    colorlinks=true,
    linkcolor=blue,
    citecolor=magenta,
    filecolor=magenta,
    urlcolor=blue
    }
\usepackage{cleveref}

\DeclareMathOperator{\im}{im}
\DeclareMathOperator{\supp}{supp}
\DeclareMathOperator{\nnz}{nnz}

\theoremstyle{plain}
\newtheorem{theorem}{Theorem}[section]
\newtheorem{lemma}[theorem]{Lemma}
\newtheorem{corollary}[theorem]{Corollary}

\theoremstyle{definition}
\newtheorem{definition}[theorem]{Definition}

\theoremstyle{remark}
\newtheorem{remark}[theorem]{Remark}

\newcommand{\F}{\mathbb{F}}

\newcommand{\cone}{{\rm cone}}
\newcommand{\Z}{\mathbb{Z}}
\newcommand{\N}{\mathbb{N}}

\newcommand{\dist}{\mathrm{dist}}
\newcommand{\SAT}{\mathbf{SAT}}

\usepackage{mdframed}

\newmdenv[
    backgroundcolor=blue!5,
    linecolor=blue!50!black
]{disclaimer}

\begin{document}
\title{Near-optimal high-rate surgery from linear PCPPs}
\author{Alexander Cowtan}
\email{alexander.cowtan@xanadu.ai}
\affiliation{Xanadu, Toronto, ON M5G 2C8, Canada}
\author{Benjamin Ide}
\affiliation{Xanadu, Toronto, ON M5G 2C8, Canada}

\maketitle
\abstract{A central problem for surgery with quantum  Low-Density Parity Check (LDPC) codes is the design of auxiliary systems which measure large sets of logical operators in parallel, while preserving sparsity and fault distance.
We introduce a very general method for designing surgery gadgets which are ‘high rate’, meaning that the number of operators measured in parallel is large in comparison to the size of the auxiliary system. Given an arbitrary initial $\llbracket n, k, d \rrbracket$ quantum LDPC Calderbank-Shor-Steane (CSS) code, and an arbitrary subcode of size $\mu \leq n$ that contains $t \leq k$ logical qubits, the method produces a sparse high-rate surgery gadget with size $\mu (\log\mu)^{\mathcal{O}(\log\log \mu)}=\mu^{1+o(1)}$ which measures all $t$ logicals in the subcode. The space overhead is asymptotically optimal up to subpolynomial factors, as the lower bound is $\Omega(\mu)$. This gadget is produced in time polynomial in $\mu$. When logical measurement is performed using these surgery gadgets, the phenomenological fault distance is at least $d$ when performed for $d$ rounds.

Our main result comes from relating surgery gadgets with relative cosystolic expansion to linear Probabilistically Checkable Proofs of Proximity (PCPPs), which allow a randomised verifier to probabilistically verify the input to a linear circuit, using only oracle access to the input and a claimed proof.}

\section{Introduction}

Quantum error correction (QEC) is vital for large-scale, fault-tolerant quantum computation~\cite{shor1995scheme,gottesman1997stabilizer,Kitaev1997qec}. One modern approach to quantum error correction relies on the high encoding rates of quantum Low-Density Parity Check (LDPC) codes~\cite{leverrier2015quantum,breuckmann2017hyperbolic,breuckmann2021quantum,panteleev2021degenerate}. However, while quantum LDPC codes can make excellent quantum memories~\cite{bhardwaj2026high, bravyi2024high}, computing with them is complicated, and can incur high overheads, in space or in time~\cite{cohen2022low,yoder2025tour}.

Lattice surgery on surface codes~\cite{Horsman2012LatticeSurgery} has been generalised to LDPC codes~\cite{cohen2022low,cowtan2024css,williamson2026low,ide2025fault}, promising low-overhead, addressable fault-tolerant quantum computation by logical measurements~\cite{litinski2019game}. LDPC code surgery works by deforming the original code into a series of other codes with enlarged stabiliser groups, thereby measuring logical operators over time.

However, managing the time- and space-overhead tradeoffs in LDPC code surgery is challenging. Past work has tended to either result in very space-efficient but time-inefficient (in the worst case) constructions~\cite{yoder2025tour,he2025extractors}, or vice versa~\cite{baspin2025fast}. 

The lowest previous asymptotic spacetime cost depends slightly on the scenario. Say we start with an $\llbracket n, k, d\rrbracket$ code. If one has an arbitrary subcode of size $\mu$ containing $t$ logicals, with weight at most $\mathfrak{w}$, and wants to measure them all simultaneously, then this can be achieved in $\tilde{\mathcal{O}}(t\mathfrak{w})$ space (on top of the original code), and $\mathcal{O}(d)$ time by `branching' and using expander graphs~\cite{zhang2025time, cowtan2026parallel}. Alternatively, they can be measured in $\mathcal{O}(\mu d)$ space~\cite{zhang2025time,zheng2025high} and $\mathcal{O}(d)$ time by `thickening'; given certain special properties of the subcode, such as soundness, this can be reduced~\cite{cowtan2025fast,zheng2025high}, but such special properties are hard to find, and hard even to compute when they exist.

For time reduction, fast surgery (that is, surgery which can be done in fewer than $\mathcal{O}(d)$ rounds) has been studied~\cite{baspin2025fast,cowtan2025fast,chang2026constant}; but this currently either comes with a formidable $\mathcal{\tilde{O}}(nkd)$ space cost~\cite{baspin2025fast} or cannot address an arbitrary choice of subcode~\cite{cowtan2025fast,chang2026constant,hillmann2025single}.

There is a clear lower bound to the space cost of measuring a subcode using surgery. To preserve the LDPC property, any auxiliary system must connect sparsely to the original subcode, and so the size of that system must scale as $\Omega(\mu)$, where $\mu$ is the size of the subcode. \footnote{This does not mean that $\Omega(\mu)$ is a lower bound on performing an equivalent logical measurement; if a smaller subcode which contains the same logicals exists, then the overhead can be reduced. Here, we assume that the subcode to be measured is dictated in advance.}
In this work, we present a procedure for measuring a subcode of size $\mu$ in $\mu^{1+o(1)}$ space cost, and $\mathcal{O}(d)$ time cost. The code remains sparse at all times, and the fault distance is at least $d$. The procedure works for any LDPC CSS code, and does not demand any special properties of the original code, or the subcode being measured. Our method is asymptotically much more efficient than using branching and expander graphs when the subcode being measured has densely packed logicals, that is $t\mathfrak{w} \gg \mu$. Our method is asymptotically more efficient than thickening unless the distance $d = \mu^{o(1)} \leq n^{o(1)}$, i.e. the distance scales subpolynomially in the size of the subcode, and therefore the initial code. For example, when measuring a subcode containing $\Omega(k)$ logicals on a good quantum LDPC code~\cite{panteleev2022asymptotically}, branching and expander graphs would yield a space cost of $\tilde{\mathcal{O}}(n^2)$, and thickening would yield a space cost of $\mathcal{O}(n^2)$. By contrast, our method would yield a space cost of $n^{1+o(1)}$.

\begin{table}[]
    \centering
    \begin{tabular}{|l|c|c|c|c|}
    \hline
        \textbf{Method} & \textbf{Ancillary space} & \textbf{Time} & \textbf{Spacetime volume} & \textbf{Good code,} $\mathbf{t = \Omega(k)}$ \\
    \hline
        Thickening~\cite{zhang2025time} & $\mathcal{O}(\mu d)$ & $\mathcal{O}(d)$ & $\mathcal{O}(d(n+d\mu))$ & $\mathcal{O}(n^3)$\\
        Gauging~\cite{williamson2026low} & $\tilde{\mathcal{O}}(\mathfrak{w})$ & $\mathcal{O}(dt)$ & $\tilde{\mathcal{O}}(td(n+\mathfrak{w}))$ & $\tilde{\mathcal{O}}(n^3)$ \\
        Parallel~\cite{cowtan2026parallel} & $\tilde{O}(t\mathfrak{w})$ & $\mathcal{O}(d)$ & $\tilde{\mathcal{O}}(d(n+t\mathfrak{w}))$ & $\tilde{\mathcal{O}}(n^3)$ \\
        Fast~\cite{baspin2025fast} & $\tilde{\mathcal{O}}(nkd)$ & $\mathcal{O}(1)$ & $\tilde{\mathcal{O}}(nkd)$ & $\tilde{\mathcal{O}}(n^3)$ \\
    \hline
        This work & $\mu^{1+o(1)}$ & $\mathcal{O}(d)$ & $\mathcal{O}(d(n+\mu^{1+o(1)})) $& $n^{2+o(1)}$ \\
    \hline
    \end{tabular}
    \caption{Comparison of new and prior methods for measuring a subcode of size $\mu$ in an arbitrary $\llbracket n,k,d \rrbracket$ qLDPC CSS code. $\mathfrak{w}$ is the weight of the largest logical to be measured in the subcode. $t$ is the number of independent logicals being measured. The $\tilde{\mathcal{O}}(\cdot)$ notation suppresses polylog factors.}
    \label{tab:comparison}
\end{table}

The central lemma in the procedure is as follows:

\begin{restatable}{lemma}{sparseexpander}
\label{lem:sparse_relative_expander}
Let $A = \F_2^\mu \rightarrow \F_2^\nu$ be a length-1 cochain complex over $\F_2$ such that $\delta^1_A: \F_2^\mu \rightarrow \F_2^\nu$ is LDPC and $\nu \in \mathcal{O}(\mu)$.

    Then there exists a length-1 cochain complex $\chi = \chi^1 \rightarrow \chi^2$ over $\F_2$, along with cochain map $g: \chi \rightarrow A$, such that:
    \begin{enumerate}
        \item $g$ has row and column weights at most 1,
        \item $\dim \chi^1 + \dim \chi^2 = \mu(\log \mu)^{\mathcal{O}(\log \log \mu)} = \mu^{1+o(1)}$,
        \item $\delta^1_\chi : \chi^1 \rightarrow \chi^2$ is LDPC,
        \item $g^1(\ker \delta^1_\chi) = \ker \delta^1_A $,
        \item $|\delta^1_\chi x| \geq \min_{z \in \ker\delta^1_\chi} |g^1(x-z)|$, $\forall x \in \chi^1$.
    \end{enumerate}
    Moreover, $\chi$ can be constructed in time polynomial in $\mu$.
\end{restatable}

Importantly, as a consequence the surgery gadget from our procedure can always be constructed in time polynomial in $\mu$. This is a major distinction from current attempts to reduce the spacetime overheads, which commonly require numerical distance calculation or estimation even in order to verify the procedure~\cite{gu2026qgpu,bhardwaj2026high,zheng2025high}, and are constructed in a somewhat ad-hoc manner.

There are some downsides of this procedure. The first is that we have not computed the constants, nor tried to optimise them, and so whether the procedure can be made practical for high-rate codes at blocklengths of interest~\cite{bhardwaj2026high, bravyi2024high} is not known to the authors. The second is that, while the code remains sparse at all times, the \textit{detectors} (sets of checks which detect spacetime faults) are not generally; that is, the quantum system is always sparse, but the decoding matrix which the classical coprocessor must decode in real time may not be. This is similar to the construction in Ref.~\cite[Sec.~IV]{baspin2025fast}, but unlike the constructions in Refs.~\cite{cowtan2026parallel,yuan2026parsimonious,ide2025fault,he2025extractors}, which are sparse and yield sparse detectors. We believe that the detectors can be made sparse without further asymptotic overhead, but it substantially complicates the procedure, which is already nontrivial. The third is that, while measuring a subcode allows one to address a large number of possible subsets of operators (as we show in Sec.~\ref{sec:coning}), including any arbitrary set of single-qubit logicals, it is not fully addressable in the sense of Ref.~\cite{cowtan2026parallel}.

Our work also opens up a new avenue of research, by connecting the construction of efficient surgery gadgets to Probabilistically Checkable Proofs of Proximity (PCPPs), a topic in complexity theory~\cite{arora2009computational, meir2012combinatorial}. Our construction essentially works by taking a linear PCPP, which produces a set of linear circuits, finding a realisation of them as a matrix, and sparsifying that until it is LDPC, then incorporating that into the auxiliary complex used for surgery~\cite{ide2025fault}. This is a circuitous route, and it would be more practical to be able to borrow from the line of reasoning that leads to linear PCPPs in the first place, and get direct matrix constructions from that; that would be much cleaner than transporting between circuits and matrices in the manner we have done here. This connection between linear PCPPs and auxiliary complexes for surgery is completely unexplored, and, as we demonstrate in this paper, it deserves further attention.

As PCPPs are quite a departure from the arguments typically used in quantum code surgery, we aim to use PCPPs as little as possible: we give the relevant definitions and use a single theorem from Ref.~\cite{meir2012combinatorial}. Other than that, reading the paper requires no prior understanding of the literature on PCPPs. In App.~\ref{app:improved_PCPP} we give a proof sketch of a slightly better linear PCPP than that which is used in the main body of the paper. This implies that the asymptotic overhead of $\mu^{1+o(1)}$ for high-rate surgery can be improved to $\mu (\log \mu)^{\mathcal{O}(1)} = \tilde{\mathcal{O}}(\mu)$. As this requires some further knowledge of PCPPs, and gives only a minor improvement to the asymptotic overheads, we do not use this argument in the main body.

\section{Preliminaries}

Very little of the preliminaries section is novel, although we frame logical bases and surgery by coning in a manner which is slightly different from past works~\cite{cowtan2026parallel,ide2025fault}.

\begin{definition}\label{def:field_distances}
    Let $K$ be a field. Let $y \in K^m$ be a vector. Then $|y|_K$ is the Hamming weight over $K$. Let $C \subseteq K^m$ be a subspace. Then $\dist_K(y, C) = \min_{c\in C}|y-c|_K$. Where the field is obvious, we just write $|y|$ and $\dist(y,C)$.
\end{definition}

Note that when $K = \F_{2^k}$, $y-c = y+c$.

\begin{definition}
    A cochain complex is a graded vector space $C = \bigoplus_{i \in \Z}C^i$ over a field $K$, equipped with $K$-linear differential $\delta = \sum_i\delta^i$ where $\delta^i: C^i \rightarrow C^{i+1}$ and $\delta\delta = 0$; on components, this means that $\delta^{i+1}\delta^i = 0$.
\end{definition}

\begin{definition}
    Let $A$ and $C$ be cochain complexes. Then a cochain map $f: A \rightarrow C$ is a linear map such that $f\delta_A = \delta_Cf$; on components, this means that $f^i \delta^{i-1}_A = \delta^{i-1}_Cf^{i-1}$.
\end{definition}

Every cochain complex in this paper satisfies $C^i = K^{N_i}$ for some $N_i \in \N$, $\forall i$, so that Hamming weights on cochain complexes are well-defined, and we can consider the differentials and cochain maps to be matrices. Also, every cochain complex has at most 4 nonzero components.

\subsection{Linear PCPPs}

A \textit{Probabilistically Checkable Proof} (PCP) verifier for a language $L$ is an algorithm that verifies whether a word $w$ is in $L$ by querying a small number of bits from a purported proof $\pi$~\cite{arora2009computational,dinur2007pcp}. Typically, the verifier is given only oracle access to $\pi$. The verifier, viewed as an algorithm, outputs its queries as a list of coordinates, and a circuit to be applied to the answers of those queries. The verifier \textit{accepts} if the circuit is satisfied by the answers given to the queries, and \textit{rejects} otherwise. The verifier may accept a word $w$ even if $w \not\in L$, but a PCP verifier has a rejection ratio which upper bounds the likelihood of erroneous acceptance happening.

A \textit{Probabilistically Checkable Proof of Proximity} (PCPP) is a generalisation of a PCP, with a slightly different setup~\cite{ben2005short, dinur2006assignment, meir2012combinatorial}. A PCPP verifier takes in two inputs: (1) an explicit input, which the verifier is allowed to read entirely, and (2) an implicit input, from which the verifier is only allowed to access a small number of bits. The aim is then to verify that the implicit input $x$ is close to the set
\[
L(w) := \{x' : (w,x')\in L\},
\]
by reading all of $w$, a small number of the bits from $x$, and from an additional proof $\pi$. \footnote{A PCP verifier is therefore a PCPP verifier where there is no implicit input.}

We are concerned with \textit{linear} PCPPs~\cite{meir2012combinatorial}. Let $K$ be a field. The explicit input of a linear PCPP is a linear subspace $W \subseteq K^m$. The verifier reads a small number of field elements from a vector $x \in K^m$, and from an additional proof string $\pi \in K^\ell$. The verifier always accepts if $x \in W$, and rejects with some probability if $x$ is far from $W$, that is if $\dist_K(x,W)$ is high. The subspace is defined by a linear circuit.

\begin{definition}\label{def:linear_circuit}
     \cite{valiant1977graph} A linear circuit is defined over a field $K$. Each wire of the circuit carries a copy of $K$, and each gate in the circuit computes a linear combination of its inputs. Every output of the circuit is a linear function of its inputs, and every linear function over $K$ can be computed by such a circuit. The size of a linear circuit is the number of wires in the circuit. Crucially, we allow only circuits with fan-in and fan-out upper-bounded by 2. \footnote{Note also that, following Ref.~\cite{meir2012combinatorial}, these circuits are strictly linear, i.e. non-affine.}
\end{definition}

We say that a linear circuit $\varphi: K^m \rightarrow K^p$ accepts an input $x \in K^m$ if $\varphi(x) = 0 \in K^p$, and otherwise we say that $\varphi$ rejects $x$. Observe that the set of inputs accepted by a linear circuit $\varphi: K^m \rightarrow K^p$ is a linear subspace of $K^m$ of dimension at least $m-p$. We denote the subspace of accepted inputs by $\SAT(\varphi)$, and
say that $\varphi$ accepts $\SAT(\varphi)$.

\begin{remark}
    We can think of $\SAT(\varphi)$ as being the `kernel' of $\varphi$, seen as a matrix. That is, if $\varphi$ is a linear map, then $\SAT(\varphi) = \ker(\varphi)$, and $\SAT(\varphi) = W \subseteq K^m$, the explicit input to the PCPP.
\end{remark}

\begin{definition}[Linear PCPP, adapted from~\cite{meir2012combinatorial,bensasson2009sound}]\label{def:linear_pcpp}
    Let $r, \ell, d: \N \rightarrow \N$, $\rho: \N \rightarrow (0,1)$. A linear PCPP verifier $V$ with \textit{randomness complexity} $r$, \textit{proof length} $\ell$, \textit{decision complexity} $d$, and \textit{rejection ratio} $\rho$, is a probabilistic polynomial time machine that satisfies the following requirements:

    \begin{enumerate}
        \item \textbf{Input:} The verifier $V$ takes as input a finite field $K$ and a linear circuit $\varphi: K^m \rightarrow K^p$ of size $n$.
        \item \textbf{Output:} For every $\omega \in \{0, 1\}^{r(n)}$, the verifier $V$ outputs a tuple $I_\omega$ of coordinates in $[m+ \ell(n)]$, and a linear circuit $\psi_\omega: K^{|I_\omega|}\rightarrow K^{q_\omega}$ of size at most $d(n)$.
        \item \textbf{Randomness complexity:} On every finite field $K$ and input circuit $\varphi$, \[\omega \leftarrow \{0, 1\}^{r(n)}\] is chosen uniformly at random.
        \item \textbf{Completeness:} For every $x \in \SAT(\varphi)$ there exists a string $\pi \in K^{\ell(n)}$ such that, for every $\omega \in \{0, 1\}^{r(n)}$,
        \[\psi_\omega ((x\circ \pi)_{\restriction I_\omega})\ \mathrm{accepts}\]
        where $\psi_\omega$ and $I_\omega$ are generated by the verifier $V$ on input $(K, \varphi)$. We refer to $\pi$ as the \textit{proof of} $x$.
        \item \textbf{Soundness:} For every $x \in K^m$ and every string $\pi \in K^{\ell(n)}$,
        \[ \mathrm{Pr}_{\omega \leftarrow \{0, 1\}^{r(n)}} \big[ \psi_\omega ((x\circ \pi)_{\restriction I_\omega})\ \mathrm{rejects}\big] \geq \frac{\rho(n)}{n}\cdot \dist(x, \SAT(\varphi)),\]
        where $\psi_\omega$ and $I_\omega$ are generated by the verifier $V$ on input $(K, \varphi)$.
    \end{enumerate}
\end{definition}
\begin{remark}\label{rem:pcpp_runtime}
Note that generating a \textit{single} circuit corresponding to a string $\omega$ is probabilistic if $\omega$ is chosen randomly, and runs in polynomial time in the input size $n$. The probability in (5.) is the probability of rejection \textit{before} an $\omega$ is generated; once an $\omega$ is generated, both $I_\omega$ and $\psi_\omega$ are deterministic. Later on, we will use the full set of $2^{r(n)}$ strings $\omega$ with length $r(n)$. That is deterministic, and runs in time exponential in $r(n)$ and polynomial in $n$.
\end{remark}

In other words, given input $(K, \varphi)$, the verifier $V$ samples a uniformly random string $\omega \in \{0,1\}^{r(n)}$. Then, conditioned on $\omega$, it deterministically outputs a tuple of coordinates $I_\omega$ in $[m + \ell(n)]$, and a linear circuit $\psi_\omega: K^{|I_\omega|}\rightarrow K^{q_\omega}$ of size at most $d(n)$. Note that the linear circuit has fan-in and fan-out bounded by 2.

\begin{remark}\label{rem:pcpp_probs}
    By definition, 
\[\mathrm{Pr}_{\omega \leftarrow \{0, 1\}^{r(n)}}\big[ \psi_\omega ((x\circ \pi)_{\restriction I_\omega})\ \mathrm{rejects}\big] = \frac{1}{2^{r(n)}}\#\{\omega \in \{0,1\}^{r(n)}: \psi_\omega((x\circ \pi)_{\restriction I_\omega})\ \mathrm{rejects}\},\]
as each $\omega$ occurs with equal probability. Note also that $|I_\omega| \leq d(n)$, as a circuit with at most $d(n)$ wires must have at most $d(n)$ input wires. \footnote{$|I_\omega|$ is the \textit{query complexity}, which we do not need to keep track of explicitly for our purposes.} $\psi_\omega$ must also have at most $d(n)$ output wires.
\end{remark}

\begin{definition}
    Let $F : \N \rightarrow \N$. We say that a linear PCPP has a minimal field size $F$ if, whenever
it is invoked on a linear circuit $\varphi$ of size $n$ over a field $K$, it is required that $|K|\geq F(n)$.
\end{definition}

\begin{theorem}[Linear PCPP Theorem]\cite[Sec.~6]{meir2012combinatorial}\label{thm:pcpp}
    Let $\varphi : K^m \rightarrow K^p$ be a linear circuit of size $n$ over a finite field $K$.
    Then there exists a linear PCPP verifier $V$ with minimal field size $\tilde{\mathcal{O}}(\sqrt{n})$ which takes as input $(K, \varphi)$, with the following parameters:
    \begin{itemize}
        \item $r(n) = \log n + \mathcal{O}((\log \log n)^2)$,
        \item $\ell(n) = n(\log n)^{\mathcal{O}(\log \log n)}$,
        \item $d(n) = \mathcal{O}(1)$,
        \item $\rho(n) = \Omega(\frac{1}{\mathrm{polylog}(n)})$.
    \end{itemize}
\end{theorem}

\begin{remark}\label{rem:minimal_field}
    As the minimal field size is $\tilde{\mathcal{O}}(\sqrt{n})$, we note that any field with size $|K| \geq n$ satisfies the condition (for sufficiently large $n$); this is important in Sec.~\ref{sec:circuit_construction} below.
\end{remark}

\begin{remark}
    Thm.~\ref{thm:pcpp} is not stated as a theorem in Ref.~\cite{meir2012combinatorial}, but is stated explicitly in Sec.~6. In Ref.~\cite{meir2012combinatorial}, it is a step used in the process of completing the proof of Thm.~2.5 and can be reconstructed by chaining together Thms.~5.1, 4.7, and then the composition Thm.~2.11\footnote{The composition theorem is a linear version of~\cite[Thm.~2.4]{ben2004robust}.} $\log \log n$ times in order. 
\end{remark}

Thm.~\ref{thm:pcpp} may seem somewhat mysterious at first glance. It does not explain what the algorithm actually does to produce the linear circuits, and the dependency on the field size is not obvious. We give some exposition of what Thm.~\ref{thm:pcpp} is doing in Appendix~\ref{app:pcpp}. From now on, for brevity we treat Thm.~\ref{thm:pcpp} as a `black box', and do not use the details of how it works beyond what has been stated in this section.

\subsection{Expansion}

\begin{definition}[Edge expansion]\label{def:edge_expansion}
    Let $\mathcal{G}(V,E)$ be an undirected (multi)graph. Let $G$ be the incidence matrix $\F_2V \rightarrow \F_2E$. $\mathcal{G}$ has edge expansion $h$ if,
    \[|G v| \geq h \min\{|v|, |\textbf{1}_V - v|\},\]
    for all $v \in \F_2V$, where $\textbf{1}_V$ is the all-1s vector in $\F_2V$.
\end{definition}

\begin{definition}[Relative cosystolic expansion]\label{def:relative_exp}
    Let $\chi$ and $Y$ be cochain complexes over $\F_2$. Let $f: \chi \rightarrow Y$ be a cochain map, and let $\tau \in \N$. Then $\chi$ has cosystolic expansion $\beta_\tau$ relative to $f$ at degree $k$ if,
    \[|\delta^k_\chi x| \geq \beta_\tau \min\{\tau, \min_{z \in \ker \delta^k_\chi}|f^k(x-z)|\}\]
    for all $x \in \chi^k$, where $|\cdot|$ is the Hamming weight.
\end{definition}

Def.~\ref{def:relative_exp} is essentially the same definition as ``modular expansion'' in Ref.~\cite[Def.~7.2]{cowtan2025fast}, albeit in different notation, and has been used in similar homological form in e.g. Refs.~\cite{yuan2026quantum, yuan2026unified, yuan2026noncss}. It is a generalisation of the relative expansion of graphs in Ref.~\cite{swaroop2026universal}, and of cosystolic expansion in the study of high-dimensional expanders~\cite{dikstein2024coboundary}.
Throughout, we will shorten ``relative cosystolic expansion'' to ``relative expansion'', for convenience.

\begin{remark}
    Seen in the right way, relative expansion looks suspiciously like the rejection probability of a linear PCPP. In particular, say we start with some vector $x$ in the domain of $f^k$, but $x \notin \ker \delta^k_\chi$. Then the condition is that for every $\pi$ with support elsewhere in $\chi^k$, $\delta^k_\chi (x+\pi)$ is sufficiently large. Compare with the linear PCPP in Def.~\ref{def:linear_pcpp}, which stipulates that the rejection probability of $\psi_\omega((x\circ \pi)_{\restriction I_\omega})$ is sufficiently large in point 5. We can therefore think of $K^m$, the implicit input, as being `like' $\supp f$, the domain of $f$, and $K^{\ell(n)}$, the proof, as being `like' $\chi^k/\supp f$, the rest of $\chi^k$.

    Of course, there are nontrivial differences: the linear PCPP is concerned with linear circuits, rather than matrices, and is only testing coordinates specified by $I_\omega$. A main aim of the later sections is to show that these differences are sufficiently superficial that relative expanders can be constructed from linear PCPPs.
\end{remark}

A common way to increase the relative expansion of some candidate $\chi$ is to `thicken', i.e. take the tensor product with a path graph of length $L$~\cite{cowtan2025fast, cohen2022low, zhang2025time, cowtan2024ssip, baspin2025fast, zheng2025high}.

\begin{lemma}\cite[Thm.~7.8]{cowtan2025fast}\label{lem:thickening}
    Let $\chi$ be a relative expander with expansion $\beta_\tau < 1$ at degree 1, with cochain map $f: \chi \rightarrow C$. Let $R$ be the cochain complex of the path graph, $\F_2V \rightarrow \F_2E$, of length $L = \lceil\frac{1}{\beta_\tau}\rceil > 1$, where $|V| = |E|+1$. \footnote{$R$ also describes a parity check matrix of the repetition code of length $L$.}
    
    Then $\chi' = R \otimes \chi$ is a relative expander with expansion $\beta'_\tau \geq 1$ at degree 1, where the cochain map $f' : \chi' \rightarrow C$ is given by $f\circ\zeta $, where $\zeta$ is the projection of any of the $L$ copies of $\chi$ in $R \otimes \chi$ onto $\chi$.
\end{lemma}

Evidently, if $L$ is large then the cost of boosting the relative expansion via thickening can be substantial. Indeed, in the worst case, $\chi$ could have expansion $\beta_\tau = 1/\tau$, in which case thickening $\tau$ times is necessary; when $\tau =d$, the distance of a code, and $\dim(\chi) = \mathcal{O}(n)$, where $n$ is the blocklength, $R\otimes \chi$ has size $\mathcal{O}(nd)$, a formidable overhead~\cite{zhang2025time,zheng2025high,cowtan2025fast}.

\subsection{Surgery by coning}\label{sec:coning}

A crucial framework for performing surgery on CSS codes is that of mapping cones~\cite{ide2025fault}, which allow one to reason about surgery using homological algebra. In short, the idea is to attach an auxiliary complex via a (co)chain map, and then prove various properties of the surgery procedure using the resultant chain complex. The logical operators that are measured are converted into stabilisers by the auxiliary complex.

\begin{definition}
    Let $C$ be a CSS code, described by a length-2 cochain complex. Then,
    \[C =C^0 \rightarrow C^1 \rightarrow C^2\]
    where $\delta^0_C: C^0 \rightarrow C^1$ is $H_X^\top$, and $\delta^1_C: C^1 \rightarrow C^2$ is $H_Z$; $H_X$ and $H_Z$ are the $X$-type and $Z$-type parity check matrices respectively.

    Then the blocklength of $C$ is $n(C) = \dim C^1$, the quantum dimension 
    \[k(C) = \dim H^1(C) = \dim \ker \delta^1_C/\im \delta^0_C = \dim H_1(C) = \dim \ker (\delta^0_C)^\top/\im (\delta^1_C)^\top,\]
    and the distance is $d(C) = \min(d^X(C), d^Z(C))$, where
    \[d^X(C) = \min_{c \in \ker\delta^1_C\backslash \im \delta^0_C}|c|,\qquad d^Z(C) = \min_{c \in \ker (\delta^0_C)^\top \backslash \im (\delta^1_C)^\top}|c|.\]
\end{definition}

Given an \textit{auxiliary complex} $A$ and cochain map $f: A \rightarrow C$, the \textit{coned code} $\cone(f)$ has components $A^{i+1}\oplus C^{i}$ and differentials \[\delta^i_{\cone(f)} = \begin{pmatrix}
    \delta^{i+1}_A & 0 \\
    f^{i+1} & \delta^i_C
\end{pmatrix}.\]
That is, for $(a,c)\in A^{i+1}\oplus C^i$,
\[\delta^i_{\cone(f)}(a,c) = (\delta^{i+1}_Aa, f^{i+1}a + \delta^i_Cc) \in A^{i+2} \oplus C^{i+1}.\]
In the coned code, some $\bar{X}$ logicals in the original code are converted into stabilisers, measuring them. We have a short exact sequence,
\begin{equation}\label{eq:ses}
    0 \rightarrow C \rightarrow \cone(f) \rightarrow A[1] \rightarrow 0,
\end{equation}
where $A[1]$ is $A$, the auxiliary complex, with all degrees shifted by 1. We can work out which logicals are measured using the Snake lemma~\cite{ide2025fault,weibel1994introduction}, expanding out Eq.~\ref{eq:ses} to a long exact sequence.

The measurement procedure (for $\bar{X}$ logicals) does the following~\cite{he2025extractors}:
\begin{enumerate}
    \item Begin in the code $C$.
    \item Initialise all new qubits in $A[1]$ in the $\ket{0}$ state.
    \item Measure all stabilisers of $\cone(f)$, including any new $X$ and $Z$ checks, and the $Z$ checks of $C$ which have now been deformed by the cone.
    \item After $r$ rounds, measure out the qubits in $A[1]$ in the $Z$ basis, concluding the logical measurement and incurring a frame error on $C$ (which is known, and can be tracked classically).
\end{enumerate}
If $d(\cone(f)) \geq d(C)$ and $r \geq d(C)$, this procedure always has phenomenological fault distance at least $d(C)$~\cite{williamson2026low}, where faults may occur on state initialisation, qubits between syndrome rounds, and qubit measurements.

\subsubsection{Union-addressable bases}

\begin{definition}[Union-addressable basis]\label{def:uab}
    Let $C = C^0 \rightarrow C^1 \rightarrow C^2$ be an $\llbracket n, k, d \rrbracket$ qubit CSS code, viewed as a length-2 cochain complex over $\F_2$. A \textit{union-addressable basis} of $C$ is:
    \begin{itemize}
        \item A basis $\mathcal{B} = \{[u_i]\}$, for $i \in [k]$, of $H_1(C)$.
        \item A paired basis $\mathcal{B}' = \{[v_j]\}$, for $j \in [k]$, of $H^1(C)$, such that $[u_i]\cdot [v_j] = \delta_{i,j}$, where $a\cdot b$ is the canonical nondegenerate bilinear form on $H_1(C) \otimes H^1(C)$.
        \item Representatives $u_i \in [u_i]$ and $v_j \in [v_j]$ such that ${\rm supp}(u_i) \cap {\rm supp}(v_j) \neq \emptyset$ iff $i = j$, where ${\rm supp}(a)$ is the support of a vector on its basis elements in $C^1$. \footnote{For a $\bar{Z}$ logical, its basis elements are in $C_1 \cong C^1$, where the isomorphism is the identity.}
    \end{itemize}
\end{definition}

This might seem like a stringent set of conditions, but in fact any CSS code possesses a union-addressable basis~\cite[App.~C]{cowtan2026parallel}, which can be acquired via Gaussian elimination, and similar definitions and results apply to stabiliser codes more generally~\cite{gottesman1997stabilizer}. 
Canonical union-addressable bases are well-known for hypergraph product codes~\cite{quintavalle2023partitioning, chang2026constant}, and have recently been constructed for lifted product codes~\cite{bhardwaj2026high, zheng2026logical, lee2026logical}. Finding such canonical bases is an interesting topic in its own right, but beyond the scope of this work.

\begin{lemma}
    Consider any subset $\mathcal{I} \subset [k]$, and let $\mathcal{U}_\mathcal{I}$ be the set of qubits in the support of $\bigcup_{i \in \mathcal{I}} v_i$; that is, we consider the union of the supports of all the $\bar{X}$ logicals labelled by elements in $\mathcal{I}$. Then,
    \[\supp(v_j) \cap \mathcal{U}_\mathcal{I} \neq \supp(v_j), \ \forall j \not\in \mathcal{I}.\]
\end{lemma}

In other words, each logical not in the chosen subset has some support outside of the union $\mathcal{U}_\mathcal{I}$.

\proof
Let $\supp(v_j) \cap \mathcal{U}_\mathcal{I} = \supp(v_j)$ for some $j \not\in \mathcal{I}$. Then ${\rm supp}(u_j) \cap {\rm supp}(v_j) \neq \emptyset$, as $[u_j]\cdot[v_j] = 1$, and so they must overlap.
But then $\exists v_i$ such that $i \in \mathcal{I}$, and ${\rm supp}(u_j) \cap {\rm supp}(v_i) \neq \emptyset$, while $i \neq j$, which is a contradiction. So $\supp(v_j) \cap \mathcal{U}_\mathcal{I} \neq \supp(v_j)$.
\endproof

\begin{corollary}\label{cor:uab_intersection}
    Let $b = \sum_{j \in \mathcal{K}} v_j+s \in C^1$ such that $\mathcal{K}\not\subseteq\mathcal{I}$ and $s \in \im \delta^0_C$ is a stabiliser. Then,
    \[\supp(b) \cap \mathcal{U}_\mathcal{I} \neq \supp(b).\]
\end{corollary}
\proof
By the same argument, $\exists k \in \mathcal{K}\backslash \mathcal{I}$ such that $[u_k]\cdot [b] = 1$, but then $\supp(u_k) \cap \supp(v_i) \neq \emptyset$ for some $v_i$ s.t. $i \in \mathcal{I}$, a contradiction.
\endproof

\begin{remark}
    Evidently Corollary~\ref{cor:uab_intersection} applies equally to $\bar{Z}$ operators. We claim that it also applies to any arbitrary set of single-qubit Pauli logical operators in the union-addressable basis, not just to all-$\bar{X}$ or all-$\bar{Z}$ operators.
\end{remark}

\begin{remark}
    While a union of supports of $\bar{X}$ logicals $\mathcal{U}_\mathcal{I}$ in a union-addressable basis cannot contain a nontrivial $\bar{X}$ logical from any other equivalence class, $\mathcal{U}_\mathcal{I}$ can contain an $X$ stabiliser, and therefore multiple representatives from the same equivalence class. See App.~\ref{app:uab} for an example.
\end{remark}

\subsubsection{Parallel surgery}\label{sec:expanders}

Consider a memory code $C$ with a set $\mathcal{I}$ of $t$ $\bar{X}$ logicals $\Lambda_i$ in a union-addressable basis. We can measure these logicals in parallel by introducing an ancilla system $A$ such that
$A^1 = \bigcup_{i \in [t]} {\rm supp}(\Lambda_i)$ and $\delta^1_A = \delta^1_C \restriction_{A^1}$, that is the $Z$-checks in the image are those incident to the qubits in $\bigcup_{i \in [t]} {\rm supp}(\Lambda_i)$.
There is an injective cochain map $f : A \rightarrow C$, such that $f^1 \simeq \begin{pmatrix}
    I \\ 0
\end{pmatrix}$ and the same for $f^2$, where $\simeq$ means ``up to row and column permutation''.

\[\begin{tikzcd}
    & A^1\arrow[r]\arrow[d, "f^1"] & A^2\arrow[d, "f^2"] \\
    C^0\arrow[r] & C^1\arrow[r] & C^2
\end{tikzcd}\]
We refer to any $A$ with chain map $f: A \rightarrow C$ such that $f$ has the above form as a \textit{subcode} of $C$. The distance of $A$ as a quantum code can be as low as 1.

The deformed code would then be ${\rm cone}(f)$. Crucially, this will measure only the $t$ nontrivial logical operators we wish to measure, and none others; as a consequence of Corollary~\ref{cor:uab_intersection}, any nontrivial logical operator which is not in $\mathcal{I}$ will have some support outside of $\im f^1$, and therefore cannot be converted into a stabiliser in the deformed code.

The problem is that this will not generally preserve the code distance. 
We therefore desire to construct an ancilla system $\chi$ which is a relative expander, with cochain map as follows,

\begin{equation}\begin{tikzcd}\label{eq:stack_diagram}
    & \chi^1\arrow[r]\arrow[d, "g^1"] & \chi^2\arrow[d, "g^2"] \\
    & A^1\arrow[r]\arrow[d, "f^1"] & A^2\arrow[d, "f^2"] \\
    C^0\arrow[r] & C^1\arrow[r] & C^2
\end{tikzcd}\end{equation}

and take ${\rm cone}(fg)$ as our coned code, i.e.

\[\cone(fg) = \begin{tikzcd}
\chi^1\arrow[r]\arrow[dr, "f^1g^1"] & \chi^2\arrow[dr, "f^2g^2"] & \\
C^0\arrow[r] & C^1\arrow[r] & C^2
\end{tikzcd}
\]
where direct sums are taken vertically.

If ${\rm im}H^1(fg)$ contains all $t$ $\bar{X}$ logicals then they will all still be measured, by the Snake lemma~\cite{ide2025fault, weibel1994introduction}. This is automatic if $f$ has the form above, and $H^1(g)$ is surjective, that is $g(\ker \delta^1_\chi) = \ker\delta^1_A$.

\begin{definition}[High-rate surgery]\cite[Def.~6]{zheng2025high}
    A surgery gadget is \textit{high rate} when it measures many logical operators simultaneously, normalised by the size of the ancilla system. That is, if $\dim \chi$ is the size of the ancilla system, then the \textit{rate} of the surgery procedure is 
    \[\frac{\dim f^1g^1(\ker \delta^1_\chi)}{\dim \chi},\]
    where $\dim \chi = \sum_{i}\dim \chi^i$.
\end{definition}

The concept of high-rate surgery pre-existed this definition~\cite{zhang2025time}. Note also that the rate of the surgery procedure does not necessarily dictate the spacetime overhead of logical operations: measuring extra stabilisers of the code by surgery does not generally help perform logical measurements; on the other hand, measuring \textit{specific} extra stabilisers of the code can speed up logical measurements by a factor of $\mathcal{O}(d)$~\cite{baspin2025fast,cowtan2025fast,chang2026constant,hillmann2025single}. In our case, we are not making use of `fast' surgery, so a more indicative metric would be
\[\frac{\dim \big(f^1g^1(\ker \delta^1_\chi)/f^1g^1(\ker \delta^1_\chi)\cap \im \delta^0_C\big)}{\dim \chi},\]
which removes stabilisers from the equation.

While any combination of logical measurements can arise from high-rate surgery, we use `high-rate surgery' colloquially to refer to the process of measuring a subcode~\cite{zheng2025high,bhardwaj2026high,zheng2026logical}, as opposed to an arbitrary parallel set of Pauli products~\cite{cowtan2026parallel}, which is substantially harder.

Assuming that the subcode $A$ to be measured is dictated in advance, $\dim f^1g^1(\ker \delta^1_\chi)$ is also predetermined, and so the only variable is $\dim \chi$, that is how expensive the ancilla system is. Given that $\dim \chi = \Omega(\dim A)$, as $g$ must be sparse for the cone to be LDPC, the question is how close we can get $\dim \chi$ to $\dim A$, while preserving fault distance of the measurement; the results of Sec.~\ref{sec:lift} demonstrate that we can always acquire
\[\frac{\dim \chi}{\dim A} = (\dim A)^{o(1)},\]
which is almost optimal, asymptotically.

\begin{remark}
    It may seem unnecessary to carry $A$ around in the definition; it does not appear in e.g. Ref.~\cite{ide2025fault}. It is helpful for our purposes, because later we can consider only $\chi$, $A$ and $g$, without incorporating any explicit information about $C$.
\end{remark}

\begin{lemma}\label{lem:distance_preservation}
    Let $\chi$ have relative expansion $\beta_\tau = 1$ at degree 1, with $\tau = d(C)$. Let $f^1g^1$ have column weight at most 1, and let $H^2(fg)$ be injective. Then,
    \[d^X({\rm cone}(fg)) \geq d(C).\]
\end{lemma}

\proof
First, as $H^2(fg)$ is injective, every $\bar{X}$ logical of ${\rm cone}(fg)$ is stabiliser-equivalent to an unmeasured $\bar{X}$ logical in $C$, by the Snake lemma. Consider an arbitrary unmeasured $\bar{X}$ logical $\Lambda$, which we consider a vector in $\ker(\delta^1_C)\backslash {\rm im}(\delta^0_C)$. We know that $|\Lambda| \geq d^X(C)$ and $|\Lambda + q| \geq d^X(C)\ \ \forall q \in \ker\delta^1_C$ such that $[q] \neq [\Lambda]$. Now, for any $x \in \chi^1$ we must ensure that $|\Lambda - f^1g^1(x)| + |\delta^1_\chi x| \geq d(C)$. 

\begin{align}
    &|\Lambda-f^1g^1(x)| + |\delta^1_\chi x| \\
    & \geq |\Lambda - f^1g^1(x)| + \min\{d(C), \min_{z \in \ker(\delta^1_\chi)} |f^1g^1(x - z)|\}. \label{eq:min_c}
\end{align}

If the first $\min$ function evaluates to $d(C)$ then we are done. Otherwise, by commutation of the diagram in Eq.~\ref{eq:stack_diagram}, $f^1g^1(z) \in \ker(\delta^1_C)$ and $[f^1g^1(z)] \neq [\Lambda]$, as otherwise $\Lambda$ would be measured. Therefore, 
\[|\Lambda - f^1g^1(z)| \geq d^X(C)\ \ \forall z \in \ker(\delta_\chi^1).\]
Hence, for the minimal $z$ in Eq.~\ref{eq:min_c},
\begin{align*}
    &|\Lambda - f^1g^1(x)| + |f^1g^1(x - z)| \\
    &= |\Lambda - f^1g^1(x-z) - f^1g^1(z)| + |f^1g^1 (x - z)| \\
    &\geq |\Lambda - f^1g^1(z)| - |f^1g^1(x-z)| + |f^1g^1 (x - z)| \\
    &= |\Lambda - f^1g^1(z)| \geq d^X(C).
\end{align*}
Then, $d^X(C) \geq d(C)$.
\endproof

\begin{corollary}\label{cor:distance_preservation}
    Let $\chi$ have relative expansion $\beta_\tau = 1$ at degree 1, let $\tau = d(C)$, and let $f^1g^1$ have column weight at most 1. Let $H^2(fg)$ be injective. Then $d({\rm cone}(fg)) \geq d(C)$.
\end{corollary}
\proof
Recall that $d({\rm cone}(fg)) := \min \{d^X({\rm cone}(fg)), d^Z({\rm cone}(fg))\}$. By \cite[Thm.~1]{ide2025fault}, $d^Z({\rm cone}(fg)) \geq d^Z(C)$, so we are done.
\endproof

That the logical measurement has phenomenological fault distance at least $d(C)$ follows by separating spacelike and timelike faults~\cite{williamson2026low,cowtan2025fast}. In this work, we are going to relax the condition that $H^2(fg)$ is injective.

\begin{restatable}{lemma}{faultdistance}\label{lem:no_gaugefix}
    Let $\chi$ have relative expansion $\beta_\tau = 1$ at degree 1 with $\tau = d(C)$, and let $f^1g^1$ have column weight at most 1. Then, regardless of the content of $H^2(\chi)$, the measurement procedure defined by $\cone(fg)$ has phenomenological fault distance $d(C)$, when the logical measurement is performed for $d(C)$ syndrome rounds.
\end{restatable}

\begin{remark}
    This lemma follows immediately from the paragraph below \cite[Thm.~1]{williamson2026low} in the Supplementary Material. Strictly speaking, the spacetime distance proof there is for the case where $\chi$ represents a graph, rather than a more general relative expander, but the proof does not require $\chi$ to be a graph, just a relative expander, so carries over straightforwardly. We prove this formally in App.~\ref{app:proofs}.
\end{remark}

The idea is that there may be new low weight spacelike logicals in $\cone(fg)$, but they cannot produce logical spacetime faults without violating the detectors introduced during the surgery procedure. That is, the deformed code $\cone(fg)$ may have $d({\rm cone}(fg)) < d(C)$; however, because the measurement procedure for $\bar{X}$ logicals initialises all new data qubits in $\ket{0}$ and measures them out in the $Z$ basis, there are sufficient detectors to ensure that the fault distance is not lowered. 

We could equivalently compute $\mathrm{coker}(\delta^1_\chi)$ and introduce a new $\chi^3 \cong \mathrm{coker}(\delta^1_\chi)$ with the projection map $\delta^2_\chi: \chi^2 \rightarrow \chi^3$, giving a set of (potentially dense) $Z$-checks such that $H^2(\chi) = 0$, and so $H^2(fg)$ is injective. Then, instead of measuring the checks at each round, leading to a dense code, infer the corresponding detectors, which span all $d(C)$ rounds, from the initial and final measurements.

\begin{remark}
    The detectors induced by the initialisation and measurement span the full $d$ rounds, and overlap substantially on each round, and so are not generally sparse. This may present a problem for fast real-time decoding, as the decoder cannot finish decoding the first round of the surgery procedure before the last round is completed, and the decoding graph will not be sparse.
\end{remark}

\begin{remark}
    We believe that the procedure can be generalised to yield sparse detectors without further asymptotic overhead. As this results in the construction becoming extremely complicated, and requires much more exposition about linear PCPPs, we leave this for further work.
\end{remark}

\section{Lifting to relative expanders}\label{sec:lift}

The aim of this section is to prove the following lemma.

\sparseexpander*

To prove Lemma~\ref{lem:sparse_relative_expander}, we will have to take a slightly circuitous route. We summarise the construction below:
\begin{itemize}
    \item Convert $\delta^1_A$ into a linear circuit over a field extension $K = \F_{2^k}$.
    \item Use a linear PCPP to produce a series of tester circuits for the circuit acquired from $\delta^1_A$.
    \item Convert those tester circuits back into a matrix.
    \item Restrict that matrix back to $\F_2$.
    \item Perform various sparsification procedures on that matrix to get a final matrix $D$ with useful expansion properties.
    \item Insert $D$ into a cochain complex to get $\chi$.
\end{itemize}

\subsection{Constructing a linear circuit from $A$}\label{sec:circuit_construction}
Our first order of business is to translate $\delta^1_A : \F_2^\mu \rightarrow \F_2^\nu$, which is a boolean matrix, into a linear circuit $\mathcal{A}$, so that we may apply Thm.~\ref{thm:pcpp} to it and acquire a PCPP verifier. Converting $\delta^1_A$ into a linear circuit over $\F_2$, with size $\mathcal{O}(\mu)$, is immediate, as $\delta^1_A$ is LDPC. However, Thm.~\ref{thm:pcpp} has a minimal field size of $\tilde{\mathcal{O}}(\sqrt{\mu})$, meaning that our field must be larger than that, and $2$ is constant, so we cannot use $\F_2$ directly.

To this end, we first regard $\delta^1_A$ as a matrix over a field extension $\F_{2^k}$, where $k = \lceil\log \mu \rceil$; the idea being that for some sufficiently large $\mu$, $2^{\lceil \log \mu \rceil} \geq \mu$ is sufficiently large for Thm.~\ref{thm:pcpp}, see Rem.~\ref{rem:minimal_field}. \footnote{If $\mu$ is not sufficiently large for a given finite field size then we can pad the field size at only constant factor cost.} We simply replace each $1_{\F_2}$ in $\delta^1_A$ with $1_K$, which (over $\F_2$) is a $k\times k$ identity matrix, letting $\{1, \alpha_2, \cdots, \alpha_k\}$ be a basis for $\F_{2^k}$. Each row of $\delta^1_A$ computes $(\delta^1_A a)_i = \sum_{j: (\delta^1_A)_{ij} = 1}a_j$, just as it would over $\F_2$, but now the sum is in $K$.  Note that $\ker_{\F_2} \delta^1_A \subseteq \ker_K \delta^1_A$, so we do not `lose' any kernel elements by this field extension, and in fact,

\begin{restatable}{lemma}{kerneleq}\label{lem:kernel_eq}\cite[Lem.~2.22]{randriambololona2014products}
    \[ \ker_K \delta^1_A = \ker_{\F_2}\delta^1_A\otimes_{\F_2}K.\]
\end{restatable}

We then convert $\delta^1_A$ into a linear circuit $\mathcal{A}$ over $K$. Each row of $\delta^1_A$ becomes a binary tree of $XOR$s over $K$, and each column is a fan-out tree of inputs. As $\delta^1_A$ is LDPC, breaking up the fan-in and fan-out to be bounded above by 2 has constant cost. Therefore $\mathcal{A}$ has size $\mathcal{O}(\mu)$, and we can think of $\SAT(\mathcal{A})$ as $\ker_K(\delta^1_A)$.

We can now apply Thm.~\ref{thm:pcpp} directly to $\mathcal{A}$. Exhausting the possibilities for $\omega$, Thm.~\ref{thm:pcpp} generates $T = 2^{r(\mu)}$ circuits $\{\psi_\omega: K^{|I_\omega|} \rightarrow K^{q_\omega}\}$. To recap, we have:
\begin{itemize}
    \item $r = \log \mu + \mathcal{O}((\log \log \mu)^2)$,
    \item $\ell = \mu (\log \mu)^{\mathcal{O}(\log \log \mu})$,
    \item $d = \mathcal{O}(1)$,
    \item $\rho = \Omega(\frac{1}{\mathrm{polylog}(\mu)})$,
    \item $T = 2^{\log \mu + \mathcal{O}((\log \log \mu)^2)} = \mu 2^{\mathcal{O}((\log \log \mu)^2)} = \mu(\log \mu)^{\mathcal{O}(\log \log \mu)}.$
\end{itemize}

\subsection{Constructing a matrix from the PCPP}
We prove a helpful lemma to convert the set of linear circuits from Thm.~\ref{thm:pcpp} back to a matrix.

\begin{lemma}[Matrix realisation]\label{lem:matrix_realisation}
    Let $K = \F_{2^k}$. Let $\mathcal{A} : K^\mu \rightarrow K^\nu$ have a linear PCPP, with $T = 2^{r(\mu)}$. Let $\SAT(\mathcal{A}) \subseteq K^\mu$. Then there exists a matrix $\hat{D} : K^\mu \oplus K^{\ell(\mu)} \oplus K^{\eta} \rightarrow K^{\vartheta}$ such that
    \begin{itemize}
        \item $\hat{D}$ has $\mathcal{O}(Td)$ rows and columns.
        \item For every $a \in K^\mu$, $\pi \in K^{\ell(\mu)}$, $g \in K^{\eta}$,
        \[|\hat{D}(a,\pi, g)| \geq \frac{\rho T}{\mu}\dist(a, \SAT(\mathcal{A})).\]
        \item For every $a \in \SAT(\mathcal{A})$, $\exists \pi, g$ such that
        \[\hat{D}(a, \pi, g) = 0.\]
    \end{itemize}
    The rows of $\hat{D}$ have weight bounded above by a constant, but the columns do not in general.
\end{lemma}

\proof
For each random string $\omega \in \{0,1\}^r$, the linear PCPP verifier reads some coordinates of $(a, \pi)$, that is a subspace $W_\omega \cong K^{|I_\omega|} \subseteq K^\mu \oplus K^{\ell(\mu)}$ and produces a linear circuit $\psi_\omega$ of size at most $d$, and with fan-in and fan-out bounded by 2 (see Def.~\ref{def:linear_circuit}).

For each circuit $\psi_\omega$, we associate a check matrix $G_\omega$ over $K$. We do this by introducing a new variable for each wire in $\psi_\omega$. Then, each gate in each circuit becomes a check: each gate represents a linear equation, which must be satisfied; for example, each XOR gate $z = x_1 + x_2$, over $K$, becomes the check $x_1+x_2+z = 0$, i.e. a row with 3 nonzero elements of $K$. Each copy gate $x=z$, performing fan-out, becomes $x+z = 0$, i.e. a row with 2 nonzero elements. As the circuit is over $K$, rather than $\F_2$, it may have other gate equations than XOR or copy, but they must all be linear and sparse.
Lastly, for each output wire $z$ of $\psi_\omega$, set $z = 0$, i.e. a weight-1 row. Each check has bounded weight, and each variable (corresponding to a column) interacts with only a constant number of checks, because $\psi_\omega$ has fan-in and fan-out bounded above by 2.

We call the input wires of $\psi_\omega$, coming from $W_\omega$, the \textit{external} wires, and the wires which are not inputs \textit{internal} wires, and set that space to be $K^{\eta}_\omega$. That is, the domain of $G_\omega$ is $W_\omega \oplus K_\omega^{\eta}$.

Then, observe that if $\psi_\omega$ accepts an input $(a, \pi)$, then there is some assignment of variables $g_\omega$ to internal wires such that $G_\omega(a, \pi, g_\omega) = 0$; simply trace through the circuit from inputs to outputs, assigning variables to internal wires according to the gates, and find that all outputs are 0. However, if $\psi_\omega$ rejects $(a, \pi)$ then
\[|G_\omega(a, \pi, g_\omega)| \geq 1,\ \forall g_\omega.\]
This is because, starting with a rejected input, to get 0 outputs and therefore not violate those checks, some gate equations must not be satisfied.

Doing this for each string $\omega$, we end up with a collection of matrices $G_\omega$; these have some common inputs, if the subspaces $W_\omega$ overlap, but also have some inputs with no overlap, being the $K^\eta_\omega$ spaces. We therefore stack the matrices together, yielding
$\hat{D} : K^\mu \oplus K^{\ell(\mu)} \oplus K^{\eta} \rightarrow K^{\vartheta},$
where $K^{\eta} = \bigoplus_\omega K^\eta_\omega$ and $K^\vartheta$ is the collection of all checks coming from gate equations.

As each circuit $\psi_\omega$ has size at most $d$, each matrix $G_\omega$ has $\mathcal{O}(d)$ rows and columns; there are $T$ such matrices, so $\hat{D}$ has $\mathcal{O}(Td)$ rows and columns. 

For a given input $(a, \pi)$ and any choice of internal variables $g = \sum_\omega g_\omega$, 
\[|\hat{D}(a, \pi, g)| \geq \#\{\omega : \psi_\omega((a, \pi)_{\restriction I_\omega})\ \mathrm{rejects}\},\]
as each $G_\omega$ corresponding to a circuit $\psi_\omega$ must yield at least one nonzero syndrome.

By Remark~\ref{rem:pcpp_probs}, 
\[\#\{\omega \in \{0,1\}^{r(\mu)}: \psi_\omega((a\circ \pi)_{\restriction I_\omega})\ \mathrm{rejects}\} = T \mathrm{Pr}_\omega\big[ \psi_\omega ((a\circ \pi)_{\restriction I_\omega})\ \mathrm{rejects}\big],\]
and $\mathrm{Pr}_\omega\big[ \psi_\omega ((a\circ \pi)_{\restriction I_\omega})\ \mathrm{rejects}\big] \geq \frac{\rho}{\mu}\dist(a, \SAT(\mathcal{A}))$, hence
\[|\hat{D}(a,\pi, g)| \geq \frac{\rho T}{\mu}\dist(a, \SAT(\mathcal{A})).\]

Lastly, for any $a \in \SAT(\mathcal{A})$, all circuits  $\psi_\omega$ accept the input $(a, \pi)$ for some choice of $\pi$, and therefore $\exists \pi,g $ such that
\[\hat{D}(a, \pi, g) = 0.\]
\endproof

We are going to make a very minor modification now. To achieve high relative expansion, we would like to get rid of the normalisation. We can do this by repeating each circuit $L := \lceil \frac{\mu}{\rho T} \rceil$ times, keeping the same external wires but adding new internal wires for each matrix. It is easy to show that now, $|\hat{D}(a,\pi, g)| \geq \dist(a, \SAT(\mathcal{A}))$. \footnote{We could achieve the same thing by thickening the complex we make right at the end by $L$, following Lemma~\ref{lem:thickening}, but doing it in this way now avoids carrying around extra variables in the subsequent steps and simplifies the diagram in Sec.~\ref{sec:the_cone}. Note also that for sufficiently large $\mu$, $\frac{\rho T}{\mu} \geq 1$ anyway, and so this step becomes unnecessary.}

The number of circuits is now $LT$. As $L < \frac{\mu}{\rho T} + 1$ by the ceiling function, $LT < \frac{\mu}{\rho} + T$, hence $LT = \mu(\log \mu)^{\mathcal{O}(\log \log \mu)}$. As $d = \mathcal{O}(1)$ each circuit has constant size, so each matrix $G_\omega$ has constant size. Therefore the number of rows and columns of $\hat{D}$ is $\mu(\log \mu)^{\mathcal{O}(\log \log \mu)}$.

\begin{remark}
    We would like to stop here: use $\hat{D}$ as a matrix in our target complex $\chi$ and show that this achieves relative expansion. However, there are two problems: first, our matrix is over $K$, not $\F_2$. Second, it is row-sparse but not column-sparse.
\end{remark}

\subsection{Restriction to $\F_2$}

We need to convert $\hat{D}$ into a matrix over $\F_2$. Recall that the basis of $K$ that we chose was $\{1, \alpha_2, \cdots, \alpha_k\}$. We would like to test words in $\F_2^\mu$, not $K^\mu$, and therefore for each copy of $K$ only test the field element $(1, 0, 0, \cdots, 0) \in K$.

We now restrict scalars from $K$ to $\F_2$. 
Let \[ z=z^{(1)}+\alpha_2z^{(2)}+\cdots+\alpha_kz^{(k)}, \qquad z^{(j)}\in\F_2^\mu.\]
By Lemma~\ref{lem:kernel_eq}, if \[ c=c^{(1)}+\alpha_2c^{(2)}+ \cdots+\alpha_kc^{(k)} \in \SAT(\mathcal{A}), \] then $c^{(j)}\in \ker_{\F_2}(\delta^1_A)$ for every $j$. 

It follows that,
\begin{restatable}{lemma}{binarysoundness}\label{eq:binary_soundness}
 \[|\operatorname{Res}_{\F_2}(\hat{D})(a,\pi, g)| \geq \dist_{\F_2}(a^{(1)},\ker_{\F_2} \delta^1_A),\]
 for all $a,\pi,g$.
\end{restatable}
which we prove in Appendix~\ref{app:proofs}. Note that the right hand side contains only $a^{(1)}$, not $a$: we are only testing against the field element $(1,0,0,\cdots,0)$ for each copy of $K$. For simplicity, from now on we will use $a$ to refer to an element of $\F_2^\mu \subseteq K^\mu$, rather than $K^\mu$; the now-redundant vector space $\F_2^{\mu(k-1)}$ is absorbed into the other vector spaces containing $\pi$ and $g$.

However, restriction of scalars doesn't generally preserve row or column weights. Each linear equation in $\hat{D}$ over $K$ becomes $k$ binary linear equations, and multiplication by an element in $K$ becomes a $k\times k$ binary matrix. Each equation can have weight as large as $\mathcal{O}(qk)$, where $q$ is the fan-in over $K$; as the fan-in is bounded, each equation has weight at most $\mathcal{O}(k)$.
Therefore $\operatorname{Res}_{\F_2}(\hat{D})$ has row weights $\mathcal{O}(k) = \mathcal{O}(\log \mu)$.

We therefore restore bounded row weight by replacing each binary equation by a set of weight-3 XORs, introducing fresh variables. The original equation is satisfied iff the new replacement set of equations is satisfied.

Let this new matrix be $D_0$. Since each equation over $K$ corresponds to $k$ binary equations, and each binary equation would have weight $\mathcal{O}(k)$, the total size of $D_0$, over $\F_2$, is increased by a factor of $\mathcal{O}(k^2) = \mathcal{O}(\log^2 \mu)$ from the size of $\hat{D}$, over $K$. Because we are only adding gates, and all gate equations can still be satisfied by the same inputs over $K$, we still have
\[|D_0(a,\pi, g)| \geq \dist_{\F_2}(a,\ker \delta^1_A)\]
and, for any $a \in \ker \delta^1_A$, $\exists \pi, g$ such that
\[D_0(a, \pi, g) = 0.\]
Since the size of $\hat{D}$ was $\mu(\log \mu)^{\mathcal{O}(\log \log \mu)}$, an additional $\mathcal{O}(\log^2 \mu)$ overhead is absorbed without changing the asymptotics.

From now on, we shall bundle together $\pi$ and $g$ for simplicity, such that $D_0 : \F_2^{\mu} \oplus P_0 \rightarrow Q_0$.

\subsection{Reducing column weight}

We have now acquired a matrix $D_0$ with useful expansion properties, but which has high column weight. It has high column weight because each linear circuit from the linear PCPP can have overlapping input bits and proof bits; the internal structure of the circuits (internal wires, internal gates) is all sparse. We therefore need a way to separate out the input bits without losing the expansion properties we have carefully acquired. To do this, we prove the following lemma.


\begin{lemma}[Expander sparsification lemma]
\label{lem:degree_reduction}
Let

\[D_0 : \F_2^{\mu} \oplus P_0 \rightarrow Q_0\]
have row weight bounded above by a constant. Let
\[|D_0(a, p)| \geq \beta_0 \dist(a, \mathcal{S})\ \forall a \in \F_2^\mu, p\in P_0,\]
for some constant $\beta_0$ and $\mathcal{S} \subseteq \F_2^\mu$. Let
\[D_0(a, p) = 0\]
for any $a \in \mathcal{S}$ and some choice of $p$.

Then there exists,
\[
    D:\F_2^{\mu}\oplus P \rightarrow Q
\]
such that:
\begin{enumerate}
    \item $D$ has bounded row and column weight;
    \item $|D(a, p)| \geq \beta\,\dist(a,\mathcal{S})$
    for all $a, p$, where $\beta>0$ is a constant dependent on $\beta_0$,
    \item For any $a \in \mathcal{S}$, $\exists p \in P$ such that
    \[D(a,p)=0.\]
    \item The size of $D$ is $\mathcal{O}\bigl(\nnz(D_0)+\mu + \dim(P_0)\bigr)$.
\end{enumerate}
\end{lemma}

\proof
Say that a variable $v$ appears in $m$ rows in $D_0$. For each row in $D_0$ that a variable $v$ appears in, introduce a new copy of $v$, along with the original copy; thus we have a set
\[V= \{v, v_1,\cdots, v_m\}\]
of variables.
Next, introduce a constant-degree expander graph with vertices $V$ and expansion $h \geq 2$. \footnote{This can be acquired by taking a constant-degree simple expander graph with constant expansion $h'$~\cite{reingold2002entropy,hoory2006expander} and then copying each edge $b$ times to make a multigraph with expansion $h = h'b$, following Def.~\ref{def:edge_expansion}. One can also add extra vertices without penalty, so long as the size is $O(|V|)$ and expansion $h\geq 2$.} Impose the equation $u = u'$ on edges $(u, u')$ of the graph. Then, for each time $v$ appears in $D_0$, replace it by a copy $v^i$; we are left with one extra $v$, which only takes part in the equations for its edges in the graph.

Do this for each variable, until the matrix has constant column weight (the row weight stays constant), satisfying (1.). We therefore have a new matrix $D:\F_2^{\mu}\oplus P \rightarrow Q$. Each original copy of $v$ with support in $\F_2^\mu \subseteq \F_2^\mu \oplus P_0$ is kept in $\F_2^\mu \subseteq \F_2^\mu \oplus P$; everything else is moved into $P$.

Now, consider an arbitrary input $(a, p)$ to $D$. For each expander graph, let $\bar{v}$ be the majority value of the input in the set $V$: that is, there is an assignment of values to each vertex, and let $\bar{v}$ be the majority value of those. \footnote{If there is a tie, choose the value of the original copy of $v$.}
Let $M$ be the total number of vertices whose values disagree with the majority value in their respective expander graphs. Every expander edge connecting vertices which disagree yields a 1 syndrome, and so the number of total violated edge checks is
\[N_{\mathrm{edge}} \geq 2M\]
by the edge expansion of each graph.

Then let $(\bar{a}, \bar{p})$ be the resulting assignment that would be input to $D_0$ by taking the majority value $\bar{v}$ for each expander graph. Flipping all minority variables to the majority value can flip at most $M$ checks in $D$, and so,
\[N_{\textrm{non-edge}} \geq |D_0(\bar{a}, \bar{p})| - M\]
where $N_{\textrm{non-edge}}$ is the number of violated checks in $D$ which aren't edge checks. Therefore,
\[|D(a, p)| \geq N_{\mathrm{edge}} + N_{\textrm{non-edge}} \geq |D_0(\bar{a}, \bar{p})| - M + 2M = |D_0(\bar{a}, \bar{p})| + M.\]
Then, $\dist(a, \bar{a}) \leq M$, because whenever an original vertex (before copying) disagrees with the majority of $V$, it is counted as one of the minority vertices in $M$.

Therefore, $\dist(a, \mathcal{S}) \leq M + \dist(\bar{a}, \mathcal{S})$, by the triangle inequality. If $\beta_0 \leq 1$, then
\begin{align*}
    &|D(a,p)| \geq \beta_0(\dist(a, \mathcal{S}) - M) + M \\
    &\geq \beta_0\dist(a, \mathcal{S}) + (1-\beta_0)M\\
    &\geq \beta_0\dist(a, \mathcal{S}).
\end{align*}
If $\beta_0 > 1$, then
\begin{align*}
    &|D(a,p)| \geq \beta_0\dist(\bar{a}, \mathcal{S}) + M  \\
    & \geq \dist(\bar{a}, \mathcal{S}) + M  \\
    &\geq \dist(a, \mathcal{S}).
\end{align*}

So $|D(a,p)| \geq \min(1, \beta_0) \dist(a,\mathcal{S})$, fulfilling (2.).

For (3.), because there exists $\bar{p} \in P_0$ which satisfies $D_0(a, \bar{p}) =0$ for any $a \in \mathcal{S}$, the same is true for $D$ and a $p \in P$. Each time an assignment of $a$ or $\bar{p} \in P_0$ encounters an expander graph, simply set all vertices to have that same value.

(4.) is satisfied because each expander graph has size on the order of the number of rows in which a variable appears, the expanders are constant-degree, and there are at most as many expanders as there are variables, hence the size of $D$ is $\mathcal{O}\bigl(\nnz(D_0)+\mu + \dim(P_0)\bigr)$.
\endproof

In our case, $\mathcal{S} = \ker \delta^1_A$, and $\nnz(D_0) = \mu(\log \mu)^{\mathcal{O}(\log \log \mu)}$; the number of rows in $D_0$ is $\mu(\log \mu)^{\mathcal{O}(\log \log \mu)}$, and each row has constant weight. Similarly, $\mu + \dim(P_0) = \mu(\log \mu)^{\mathcal{O}(\log \log \mu)}$, and so the size of $D$, our sparsified matrix, is $\mathcal{O}\bigl(\nnz(D_0)+\mu + \dim(P_0)\bigr) = \mu(\log \mu)^{\mathcal{O}(\log \log \mu)}$. Also, $\beta_0 = 1$ in our case, and so
\[|D(a,p)| \geq \dist(a,\ker\delta^1_A)\]
for any $a \in \F_2^\mu, p \in P$.

\begin{remark}
    This step can be viewed as weight reduction in the style of Hastings~\cite{hastings2016weight,hastings2023quantum,sabo2024weight}, but crucially the property we are aiming to preserve is the relative expansion of the code, rather than its distance. It is also similar to the `expander-replacement' transformation from Ref.~\cite[Lem.~3.2]{dinur2007pcp}, taken from Ref.~\cite{papadimitriou1991optimization}.
\end{remark}

\subsection{The mapping cone}\label{sec:the_cone}

We can now prove Lemma~\ref{lem:sparse_relative_expander}, restated below for convenience. 

\sparseexpander*
\proof
Insert the $D: \F_2^\mu \oplus P \rightarrow Q$ constructed in the previous sections into the following cochain complex:
\[\chi = \begin{tikzcd}
    P \arrow[r, "D^{(2)}"]& Q \\
    \F_2^\mu \arrow[ur, "D^{(1)}"] \arrow[r, "\delta^1_A"'] & \F_2^\nu
\end{tikzcd}
\]
where direct sums are taken vertically, and $D^{(1)} = D_{\restriction \F_2^\mu}$, $D^{(2)} = D_{\restriction P}$. That is,
\[\delta^1_\chi(a,p) = (\delta^1_A a, D(a,p)).\]

\begin{enumerate}
    \item Let $g: \chi \rightarrow A$ be the projection of $\chi$ onto $A$. (1.) is satisfied.
    \item As $D$ has size $\mu(\log \mu)^{\mathcal{O}(\log \log \mu)}$, and $\dim \F_2^\nu \in \mathcal{O}(\mu)$, (2.) is satisfied.
    \item Both $D$ and $\delta^1_A$ are sparse, so (3.) is satisfied.
    \item By property (3.) of Lemma~\ref{lem:degree_reduction}, setting $\mathcal{S} = \ker \delta^1_A$, $g^1(\ker \delta^1_\chi) = \ker\delta^1_A$, so (4.) is satisfied.
    \item As $|D(a,p)| \geq \dist(a,\ker\delta^1_A)$, (5.) is satisfied.
\end{enumerate}
\endproof

The overall mapping cone on an initial memory code and choice of subcode $A$ with injection $f: A \rightarrow C$ and size $\mu$ is then $\cone(fg)$. As $|\delta^1_\chi x| \geq \min_{z \in \ker\delta^1_\chi} |g^1(x-z)|$, $\forall x \in \chi^1$ and $f$ has column weight at most 1, 
\[|\delta^1_\chi x| \geq \min\{d(C), \min_{z \in \ker\delta^1_\chi} |f^1g^1(x-z)|\}\ \ \forall x \in \chi^1.\]

By Lemma~\ref{lem:no_gaugefix}, the logical measurement then has fault distance $d(C)$ when run for $d(C)$ rounds.

The size of our ancilla system $\chi$ is $\mu(\log \mu)^{\mathcal{O}(\log \log \mu)} = \mu e^{\mathcal{O}((\log\log \mu)^2)} = \mu e^{o(\log \mu)} = \mu \mu^{o(1)} = \mu^{1+o(1)}$.

That all the operators in $A$ are measured can be checked using the Snake lemma~\cite{ide2025fault,weibel1994introduction}.

As one can always choose $A$ according to a union-addressable basis, see Def.~\ref{def:uab}, any set of single qubit all-$\bar{X}$ logical operators (in a union-addressable basis) on any LDPC CSS code can be measured in space $\mu^{1+o(1)} \leq n^{1+o(1)}$ and time $\mathcal{O}(d)$. By inverting the construction to use chain complexes, rather than cochain complexes, it works equally for all-$\bar{Z}$ logical operators. We claim that it also works immediately for heterogeneous mixtures of single-qubit $\{\bar{X},\bar{Y},\bar{Z}\}$ logicals, by Ref.~\cite{yuan2026noncss}, but do not prove it here, as it requires exposition on symplectic algebra.

Depending on whether there exist alternative union-addressable bases for other (sub)sets of such operators, our construction can also be used to measure Pauli products in the same codeblock.

\begin{remark}
    Our method can be seen as a substantial generalisation of thickening. The construction of $\chi$ is very similar to that of thickening, in that a new vector space is used to ensure sufficient support is generated from every $\delta^1_\chi x$, but rather than each layer in the tensor product $R \otimes A$ being a copy of $A$, there is only one layer beyond the initial copy of $A$. That layer is not the same as $A$, and it is sufficient on its own.
\end{remark}

\subsection{Construction runtime}

We now check that the construction of $\chi$ from any $A$ can be done in time polynomial in $\mu$.

\begin{enumerate}
    \item Constructing $\mathcal{A}$ from $A$ can be done in time $\mathcal{O}(\nnz(A)) = \mathcal{\tilde{O}}(\mu)$, where we allow polylog factors for arithmetic in $K = \F_{2^k}$, as $k = \log \mu$.
    \item Constructing a linear circuit from the PCPP theorem takes polynomial time, by definition (see Remark~\ref{rem:pcpp_runtime}). Enumerating the strings $\omega$ takes time $T$, which is in $\mu^{1+o(1)}$, and so getting all of the linear circuits $\{\psi_\omega\}$ takes polynomial time.
    \item Converting $\{\psi_\omega\}$ into $\hat{D}$ can be done in polynomial time by converting each gate in the circuits into checks one-by-one.
    \item Restriction of scalars takes polynomial time, and subsequently sparsifying the rows of $\hat{D}$ over $\F_2$ can be done in time linear in the number of non-zeros being replaced, which is evidently polynomial in $\mu$.
    \item Constant-degree expanders are polynomial-time constructible~\cite{reingold2002entropy,hoory2006expander}, and the total number of non-zeros in $D_0$ is $\mu^{1+o(1)}$, so the runtime for column weight reduction is polynomial in $\mu$.
\end{enumerate}

\section{Discussion and outlook}

In this work, we derived an extremely asymptotically efficient method of performing high-rate surgery. This worked by taking existing constructions of linear PCPPs and converting them into sparse complexes with relative expansion, opening up a new research program connecting linear PCPPs and surgery gadgets.

There are several further directions which we would like to explore. First, to make the construction practical we must determine and optimise the constants involved; in search of the best possible asymptotic results we ignored constants, and performed manipulations (such as field extension and reduction) which are cheap asymptotically but likely come with burdensome overheads for finite-length codes. It seems likely that, rather than taking linear PCPPs verbatim and using them to construct matrices, it would be preferable to follow the existence proofs of linear PCPPs and take those as inspiration for a more direct construction of relative expanders, avoiding linear circuits and field extensions entirely.

It may also be the case that there are linear PCPPs which are preferable for practical blocklengths, despite having worse asymptotic overheads. Ref.~\cite{meir2012combinatorial} gives constructions of several linear PCPPs other than that of Thm.~\ref{thm:pcpp}. For example, Ref.~\cite[Thm.~7.7]{meir2012combinatorial} has worse asymptotic overheads (e.g. decision complexity $\tilde{\mathcal{O}}(\mu^{2/3})$ as opposed to $\mathcal{O}(1)$) but is significantly less complicated, and likely has superior constants.

A limitation of our method for practical decoding is that the detectors are dense, and extend for $d$ rounds in time; this does not prohibit the fault distance from being $d$, but it means that a decoder cannot finish decoding the first round of the logical measurement until the last is completed, and decoding a dense decoding graph may be challenging for the state-of-the-art real-time decoders~\cite{muller2025improved,bhardwaj2026high,ye2026real}. We believe that this can be solved by using the structure of the linear PCPPs, but the proof requires carefully tracing through the proofs in~\cite{meir2012combinatorial,meir2009combinatorial}, and will greatly complicate the procedure. For brevity, we have used Thm.~\ref{thm:pcpp} as a black box in this work, and aimed to avoid substantial exposition of the specific linear PCPP construction. Aside from this, reasonable practical heuristics to generate sparse bases exist~\cite{zheng2025high}, unlike generating good relative cosystolic expanders; and verifying whether a basis is sparse is easily achieved in polynomial time, whereas verifying that a cochain complex is an expander -- or, similarly, verifying the code distance or cleaning properties directly -- is extremely difficult. So we claim that finding sparse checks to make $\chi$ exact at degree 2 is less of a concern than the design of relative expanders.

Assuming that one can find instantiations of our method with practical constants, architectural proposals which use high-rate surgery~\cite{bhardwaj2026high} can merely swap out existing constructions for new ones derived from linear PCPPs. However, it would be interesting to co-design entirely new architectures based on LDPC codes, automorphism gates and high-rate surgery, generalising extractor-style architectures~\cite{he2025extractors,yoder2025tour,webster2026pinnacle}; previously, gadgets which measured many logicals on the same codeblock simultaneously were used sparingly, as their construction was expensive to compute, and not guaranteed to yield efficient gadgets. Given that our construction runs in polynomial time, future high-rate surgery gadgets can be found efficiently, and the most practical of these chosen for use in the architecture. To build an architecture without needing extractor-style gadgets, one should first generalise to measuring non-CSS subcodes, which can include a set of $\{X,Y,Z\}$ operators; by the results of Ref.~\cite{yuan2026noncss}, this is immediate. Then one must further generalise our construction to measure a wider class of Pauli products, as opposed to just subcodes, and potentially generalise `bridges'~\cite{cross2025improved} or `adapters'~\cite{swaroop2026universal}. This seems like a difficult task, and we are not optimistic that this can be achieved in general with low asymptotic overhead, but co-design may enable discovery of adapters for high-rate surgery with practical codes. Interestingly, we may arrive at a situation similar to low-rate codes, such as surface codes~\cite{litinski2019game}, where addressable logical single-qubit measurements can be accomplished cheaply (using subcodes from union-addressable bases, in the high-rate case), but addressable products are more expensive; this is a very different regime from extractors~\cite{he2025extractors}, where a logical product measurement on a codeblock has the same cost as a logical single-qubit measurement.

The last generalisation which would be important is to fast surgery~\cite{cowtan2025fast,chang2026constant,baspin2025fast}, which performs a set of logical measurements in $O(1)$ rounds. Fast surgery measures a distance $d$ subcode, and puts conditions on the auxiliary complex $\chi$, which must itself have systolic distance $d$ at degree 1 (in our conventions). We do not believe the construction of $\chi$ presented literally in this work satisfies these conditions, but we suspect that a synthesis of the techniques here and in~\cite[Sec.~IV B]{baspin2025fast}, replacing the expansion boosting by thickening with the expansion boosting from the linear PCPP, would suffice.

While it would not have immediate implications for efficient surgery, it is also interesting to consider whether one can go the other way: take constructions of relative expanders and use them to derive linear PCPPs. We leave this for future work.

\section{Acknowledgements}
A.C. thanks Zhiyang (Sunny) He, Dominic J. Williamson and Theodore J. Yoder for many helpful discussions over the past 2 years.

Portions of this work were carried out within an Agentic ResearCH environment (ARCH) deployed at Xanadu. The ARCH couples a frontier language model (via the Claude Agent SDK) to internal code repositories, experiment and metrology databases, and literature search APIs, together with an isolated Linux environment in which it installs dependencies, executes code, and reproduces numerical results. Given a task specification, the ARCH plans and executes multi-step workflows autonomously, without step-by-step human direction. It was used here to refine and improve a proof strategy proposed by the authors connecting relative cosystolic expansion to PCPPs. All proofs were either developed or checked by the authors, who take responsibility for any errors, and the paper was written entirely by the authors.

\printbibliography[heading=bibintoc]

\appendix

\section{Construction of the linear PCPP in Thm.~\ref{thm:pcpp}}\label{app:pcpp}

Here, we give a terse explanation of Thm.~\ref{thm:pcpp}, omitting many important subtleties concerning `block access' or `row/column access'. The theorem starts originally from Ref.~\cite[Sec.~7]{meir2009combinatorial}, which, when given an input circuit $\phi$ generates a series of `tester circuits' $\{\psi\}$ that try to verify whether an input is in $\SAT(\phi)$. These do not generally have any soundness properties, the circuits may be large, and there may be many of them.
However, they have the property that for an input choice of variables $x \not\in \SAT(\phi)$, at least one tester circuit will output a 1.

These circuits are then `robustised'. This works by encoding the input space to each $\psi$ into an error-correcting code, and modifying the circuits accordingly. The purpose of this step is to make it so that, rather than just a tester circuit being dissatisfied, if the variable $x \not\in \SAT(\phi)$ is \textit{far} from $\SAT(\phi)$, the tester circuit will also be far from being satisfied. That is, there exists a linear PCPP $V'$ which has a slightly different soundness property, namely that for every input and proof string there exists \textit{some} $\omega$ for which
\[\dist\big((x\circ \pi)_{\restriction I_\omega}, \SAT(\psi_\omega)\big) \geq \rho \dist(x,\SAT(\phi)) \footnote{These distances are normalised, unlike those in Def.~\ref{def:field_distances} and throughout the main body.}\]
where $\psi_\omega$ is the linear circuit produced by $V'$, given $\omega$.

This is then used as Ref.~\cite[Lem.~5.6]{meir2012combinatorial}, which is combined with~\cite[Thm.~5.5]{meir2012combinatorial}. Very briefly, the idea is that the robustised circuits of Lem.~5.6, along with the proof string, can be combined with a \textit{simultaneous} PCPP (SPCPP), which tests many circuits in parallel and if any of the circuits are sufficiently dissatisfied then the SPCPP verifier rejects with some non-negligible probability. The SPCPP uses Reed-Solomon codes~\cite{reed1960polynomial}, which have a minimal field size, and so the resultant theorem (\cite[Thm.~5.1]{meir2012combinatorial}) has a minimal field size:

\begin{theorem}\cite[Thm.~5.1]{meir2012combinatorial}
    There exists a linear PCPP with randomness complexity $\frac{1}{2}\log n + \mathcal{O}(\log \log n)$, decision complexity $\tilde{\mathcal{O}}(\sqrt{n})$, minimal field size $\tilde{\mathcal{O}}(\sqrt{n})$ and rejection ratio $\Omega(1)$.
\end{theorem}

The decision complexity is notably worse than that of Thm.~\ref{thm:pcpp}; to resolve this, the linear PCPP of the above theorem is then `robustised' further, which encodes the input further into another error-correcting code (Thm.~4.7), this time using tensor codes~\cite{macwilliams1977theory}. Finally, this PCPP verifier is composed with itself $\log \log n$ times, producing Thm.~\ref{thm:pcpp}, which has improved decision complexity, although its other parameters are slightly worse.

\section{Example of a union-addressable basis}\label{app:uab}

Consider a $\llbracket 5,2,1 \rrbracket$ code $C$ with checks,
\[X_1X_2, \qquad X_1X_3X_5, \qquad Z_1Z_2Z_3Z_4.\]
Choose $\bar{X}$ logical representatives,
\[v_1 = X_1X_3, \qquad v_2 = X_2X_4\]
and $\bar{Z}$ logical representatives,
\[u_1 = Z_3Z_5, \qquad u_2 = Z_4,\]
which together satisfy Def.~\ref{def:uab}. 

Let $\mathcal{I} = \{1,2\}$. Then $\mathcal{U}_\mathcal{I} = \supp(X_1X_3) \cup \supp(X_2X_4) = \{1,2,3,4\}$.

Now, we can see that there is a stabiliser $X_1X_2$ contained within the union of supports (but not either of the individual logical operators), and therefore there are two $\bar{X}$ logicals of each of the equivalence classes $[v_1]$ and $[v_2]$ contained as well: $X_1X_3$ and $X_2X_3 \in [v_1]$, $X_2X_4$ and $X_1X_4 \in [v_2]$.

\section{Proofs omitted from the main body}\label{app:proofs}

\binarysoundness*
\proof
First, $|v|_{\F_2} \geq |v|_K$ for any $v$. Now,
$|\operatorname{Res}_{\F_2}(\hat{D})(a,\pi, g)|_{\F_2} \geq |\hat{D}(a,\pi, g)|_K \geq \dist_{K}(a, \SAT(\mathcal{A}))$. We have established that $\SAT(\mathcal{A}) = \ker_K\delta^1_A$. 

It remains to show that $\dist_{K}(a, \ker_K\delta^1_A) \geq \dist_{\F_2}(a^{(1)}, \ker_{\F_2}\delta^1_A)$ for any $a \in K^\mu$. This follows from Lemma~\ref{lem:kernel_eq}:
\[\ker_K \delta^1_A = \ker_{\F_2}\delta^1_A \otimes_{\F_2} K,\]
which means that 
\[\ker_K \delta^1_A = \{c^{(1)} + \alpha_2c^{(2)} + \alpha_3c^{(3)} + \cdots + \alpha_k c^{(k)}: c^{(j)} \in \ker_{\F_2} \delta^1_A\}.\]

So if $c \in \ker_K \delta^1_A$, its first coefficient $c^{(1)} \in \ker_{\F_2} \delta^1_A$. Take an arbitrary choice of $c$, and $a \in K^\mu$. Then $\supp(a^{(1)}-c^{(1)}) \subseteq \supp(a-c)$. Therefore $|a - c| \geq |a^{(1)}-c^{(1)}|$; minimising $c$ over $\ker_{K}\delta^1_A$ gives $\dist_{K}(a, \ker_K\delta^1_A) \geq \dist_{\F_2}(a^{(1)}, \ker_{\F_2}\delta^1_A)$.
\endproof

\faultdistance*
\proof
We are going to essentially describe a `stability experiment'~\cite{gidney2022stability} algebraically. We start with some preamble about fault complexes~\cite{hillmann2025single,cowtan2025fast}. A fault complex is a (co)chain complex with 4 terms:
\[F^0 \rightarrow F^1 \rightarrow F^2 \rightarrow F^3,\]
which describes a spacetime volume in the phenomenological model. $F^0$ and $F^3$ correspond to $X$ and $Z$ detectors respectively, by convention; these detectors are products of measurements (and potentially state initialisations) which must be zero if there are no errors. $X$ detectors are products of $X$ checks (and potentially $\ket{+}$ state initialisations) which detect $Z$ Pauli faults on qubits and $X$ check faults. $Z$ detectors are the same, but dualised.
$F^1$ is a set of locations at which $Z$ faults can occur ($Z$ Pauli faults or $X$ check faults), and $F^2$ is a set of locations at which $X$ faults can occur. Correspondingly, $H^2(F)$ describes equivalence classes of logical $\bar{X}$ faults, and $H_1(F)$ describes equivalence classes of logical $\bar{Z}$ faults. 

The fault distance of the fault complex is
\[d(F) = \min\{d^X(F), d^Z(F)\}\]
where $d^X(F) = \min_{f \in \ker\delta^2_F\backslash \im \delta^1_F}|f|$ and $d^Z(F) = \min_{f \in \ker(\delta^0_F)^\top\backslash \im (\delta^1_F)^\top}|f|$.

Given a CSS code $C$, we can make a fault complex describing the spacetime volume in which $C$ is initialised (in the $X$ or $Z$ basis), has syndrome extraction performed for $m$ rounds, and is then measured out (in the same basis). To do so, let 
\[M = \begin{pmatrix}
1 & 0 & 0 & \cdots & 0 \\
1 & 1 & 0 & \cdots & 0 \\
0 & 1 & 1 & \cdots & 0 \\
\vdots & \vdots & \vdots & \ddots & \vdots \\
0 & 0 & 0 & \cdots & 1
\end{pmatrix}:\F_2^{m}\rightarrow \F_2^{m+1}\]
be the incidence matrix of the path graph of length $m+1$, mapping from edges to vertices.

Then let
\[R = R^0 \rightarrow R^1\]
where we now distinguish between two cases, which we will track: (1) $\delta^0_R = M$ in order to initialise and measure out the spacetime volume in the $Z$ basis, and (2) $\delta^0_R = M^\top$ in order to initialise and measure out in the $X$ basis. \footnote{That this is true can be checked by expanding out and seeing that in cases (1) \& (2) there are $Z$ \& $X$ detectors at either end of the volume, respectively.} We test against both cases to ensure that there are no logical faults of either Pauli type which can enter the volume in the timelike direction and result in lower weight logical faults. In both cases, the spacetime volume is defined as,
\[R \otimes C = \begin{tikzcd}
	& {R^1\otimes C^0} & {R^1\otimes C^1} \\
	{R^0\otimes C^0} &&& {R^1\otimes C^2} \\
	& {R^0\otimes C^1} & {R^0\otimes C^2}
	\arrow[from=1-2, to=1-3]
	\arrow[from=1-3, to=2-4]
	\arrow[from=2-1, to=1-2]
	\arrow[from=2-1, to=3-2]
	\arrow[from=3-2, to=1-3]
	\arrow[from=3-2, to=3-3]
	\arrow[from=3-3, to=2-4]
\end{tikzcd}.\]

The equivalence classes of logical faults can be determined using the K{\"u}nneth formula~\cite{weibel1994introduction}, and the phenomenological fault distance can be computed using results on tensor products of (co)chain complexes~\cite{zeng2020minimal}. Logical faults can be generated as linear combinations of `spacelike' faults (arising from $C$) and `timelike' faults (arising from $R$), along with spacetime stabilisers. Let $m = d(C) + 3$. By~\cite[Thm.~17]{zeng2020minimal}, $d(R\otimes C) \geq d(C)$, in both cases. \footnote{Having $m \geq d(C)$ is sufficient for the distance bound, but we set $m = d(C)+3$ in order to have extra timesteps at the beginning and end, before and after the surgery procedure.} 

So far, we just have the idling volume. We will now construct, from that volume, the volume for the surgery. Given $\chi$ and $fg: \chi \rightarrow C$, we first acquire a fault complex $P \otimes \chi$, where $P = P^0 \rightarrow P^1$ has $\delta^0_P : \F_2^{d(C)} \rightarrow \F_2^{d(C)+1}$, the incidence matrix of the path graph, mapping from edges to vertices; the same as $M$ but with different dimensions.
That is, $P = P^0 \rightarrow P^1 = \F_2^{d(C)} \rightarrow \F_2^{d(C)+1}$. Then,
\[P \otimes \chi = {P^0\otimes \chi^1} \rightarrow {P^1\otimes \chi^1}\oplus {P^0\otimes \chi^2} \rightarrow {P^1\otimes \chi^2}.\]
This does not yet describe the ancilla system in spacetime; its degrees are incorrect, and it is not connected to the initial idling volume.

Define $\iota: P \rightarrow R$ as the coordinate injection with components:
\[\iota^0 = \begin{cases}
    \begin{pmatrix}
        0 & I_{d(C)} & 0 & 0
    \end{pmatrix}^\top, & \text{if case (1)}.\\
    \begin{pmatrix}
        0 & 0 & I_{d(C)} & 0 & 0
    \end{pmatrix}^\top, & \text{if case (2)}.
  \end{cases}\]
\[\iota^1 = \begin{cases}
    \begin{pmatrix}
        0 & I_{d(C)+1} & 0 & 0
    \end{pmatrix}^\top, & \text{if case (1)}.\\
    \begin{pmatrix}
        0 & I_{d(C)+1} & 0
    \end{pmatrix}^\top, & \text{if case (2)}.
  \end{cases}\]
That $\iota$ commutes as a cochain map, that is $\iota \delta_P = \delta_R\iota$, can be checked by matrix multiplication.

Then define the cochain map $\Upsilon = \iota \otimes fg : P \otimes \chi \rightarrow R\otimes C$. $\Upsilon$ satisfies the following commutative diagram:
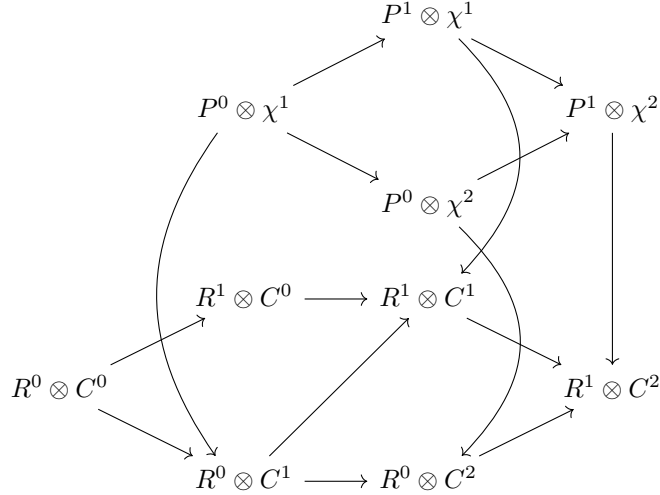
\[\begin{tikzcd}
	&& {P^1\otimes \chi^1} & \\
	& {P^0\otimes \chi^1} && {P^1\otimes \chi^2} \\
	&& {P^0\otimes \chi^2} \\
	& {R^1\otimes C^0} & {R^1\otimes C^1} \\
	{R^0\otimes C^0} &&& {R^1\otimes C^2} \\
	& {R^0\otimes C^1} & {R^0\otimes C^2}
	\arrow[from=1-3, to=2-4]
	\arrow[shift left=2, curve={height=-30pt}, from=1-3, to=4-3]
	\arrow[from=2-2, to=1-3]
	\arrow[from=2-2, to=3-3]
	\arrow[shift right=2, curve={height=30pt}, from=2-2, to=6-2]
	\arrow[from=2-4, to=5-4]
	\arrow[from=3-3, to=2-4]
	\arrow[shift left=2, curve={height=-30pt}, from=3-3, to=6-3]
	\arrow[from=4-2, to=4-3]
	\arrow[from=4-3, to=5-4]
	\arrow[from=5-1, to=4-2]
	\arrow[from=5-1, to=6-2]
	\arrow[from=6-2, to=4-3]
	\arrow[from=6-2, to=6-3]
	\arrow[from=6-3, to=5-4]
\end{tikzcd}\]

Then, the spacetime volume describing the surgery is given by $\cone(\Upsilon)$:

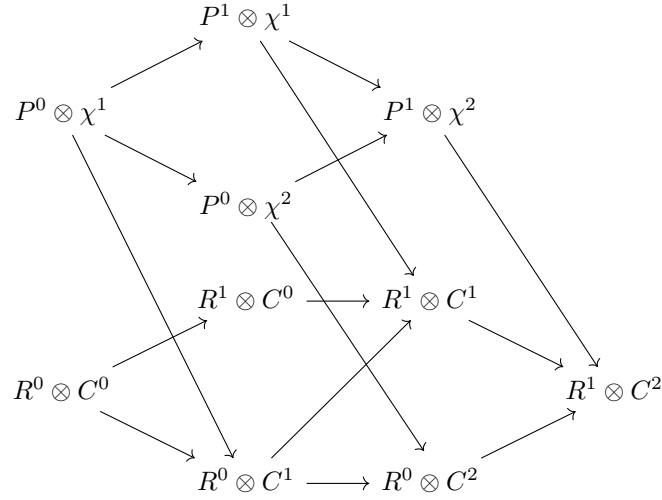
\[\begin{tikzcd}
	& {P^1\otimes \chi^1} && \\
	{P^0\otimes \chi^1} && {P^1\otimes \chi^2} \\
	& {P^0\otimes \chi^2} \\
	& {R^1\otimes C^0} & {R^1\otimes C^1} \\
	{R^0\otimes C^0} &&& {R^1\otimes C^2} \\
	& {R^0\otimes C^1} & {R^0\otimes C^2}
	\arrow[from=1-2, to=2-3]
	\arrow[from=1-2, to=4-3]
	\arrow[from=2-1, to=1-2]
	\arrow[from=2-1, to=3-2]
	\arrow[from=2-1, to=6-2]
	\arrow[from=2-3, to=5-4]
	\arrow[from=3-2, to=2-3]
	\arrow[shift left=2, from=3-2, to=6-3]
	\arrow[from=4-2, to=4-3]
	\arrow[from=4-3, to=5-4]
	\arrow[from=5-1, to=4-2]
	\arrow[from=5-1, to=6-2]
	\arrow[from=6-2, to=4-3]
	\arrow[from=6-2, to=6-3]
	\arrow[from=6-3, to=5-4]
\end{tikzcd}\]

Observe that new qubits in the ancilla system are initialised in $\ket{0}$, and measured out in the $Z$-basis.
Equivalence classes of $\bar{X}$ logical faults are given by elements in $H^2(\cone(\Upsilon))$, and $\bar{Z}$ logical faults by elements in $H_1(\cone(\Upsilon))$.

Before analysing the faults, we make a modification to the fault complex: add a set of new $Z$ detectors $\Xi^4 \cong H^3(P\otimes \chi) = \mathrm{coker}(\delta^2_{P\otimes \chi})$ to $P\otimes \chi$, adjoining to make 
\[\Xi = P^0\otimes \chi^1 \rightarrow P^1\otimes \chi^1 \oplus P^0 \otimes \chi^2 \rightarrow P^1\otimes \chi^2 \rightarrow \Xi^4,\]
such that $H^3(\Xi) = 0$. The map $\delta^3_\Xi$ may not be sparse. Note that these are just \textit{detectors}, not checks, and so the quantum code remains sparse at all times; it is only the detector matrix which may become dense.

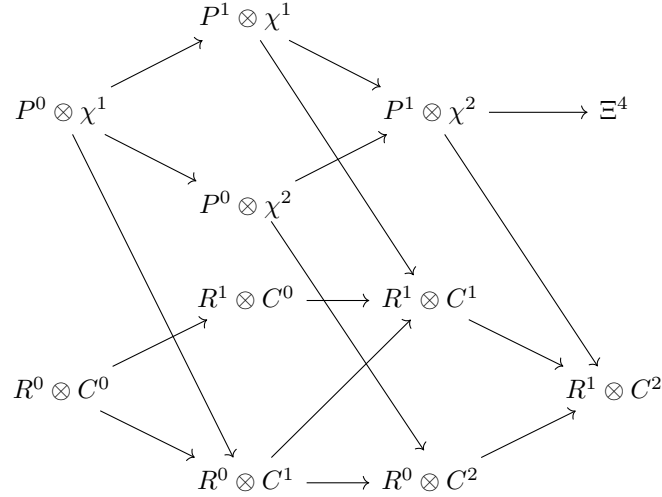
\[\begin{tikzcd}
	& {P^1\otimes \chi^1} && \\
	{P^0\otimes \chi^1} && {P^1\otimes \chi^2} & {\Xi^4} \\
	& {P^0\otimes \chi^2} \\
	& {R^1\otimes C^0} & {R^1\otimes C^1} \\
	{R^0\otimes C^0} &&& {R^1\otimes C^2} \\
	& {R^0\otimes C^1} & {R^0\otimes C^2}
	\arrow[from=1-2, to=2-3]
	\arrow[from=1-2, to=4-3]
	\arrow[from=2-1, to=1-2]
	\arrow[from=2-1, to=3-2]
	\arrow[from=2-1, to=6-2]
	\arrow[from=2-3, to=2-4]
	\arrow[from=2-3, to=5-4]
	\arrow[from=3-2, to=2-3]
	\arrow[shift left=2, from=3-2, to=6-3]
	\arrow[from=4-2, to=4-3]
	\arrow[from=4-3, to=5-4]
	\arrow[from=5-1, to=4-2]
	\arrow[from=5-1, to=6-2]
	\arrow[from=6-2, to=4-3]
	\arrow[from=6-2, to=6-3]
	\arrow[from=6-3, to=5-4]
\end{tikzcd}\]
Physically, these detectors come from the combination of $\ket{0}$ initialisations and $Z$-basis measurements of qubits in $\chi$ at the beginning and end of the surgery procedure, and have the form of products of `phantom' $Z$-checks in time.

To establish the equivalence classes of logical faults, we use the Snake lemma~\cite{ide2025fault,weibel1994introduction}, which implies:
\begin{align*}
    \dim H^i(\cone(\Upsilon)) = \dim H^i(R\otimes C) - \dim \im H^i(\Upsilon) + \dim H^{i+1}(\Xi) - \dim \im H^{i+1}(\Upsilon).
\end{align*}

\textbf{$\bar{Z}$ faults:}
We have 
\begin{align*}
    \dim H_1(\cone(\Upsilon)) &= \dim H^1(\cone(\Upsilon)) \\
    &= \dim H^1(R\otimes C) - \dim \im H^1(\Upsilon) + \dim H^{2}(\Xi) - \dim \im H^{2}(\Upsilon)\\
    &= \dim H^1(R\otimes C) - \dim \im H^1(\Upsilon) + \dim H^1(\chi) - \dim \im H^{2}(\Upsilon)
\end{align*}
where $\dim H^2(\Xi) = \dim H^0(P) \dim H^{2}(\chi) + \dim H^1(P)\dim H^1(\chi) = \dim H^1(\chi)$. The contribution of $\dim H^1(\chi)$ is due to new timelike logical faults, which can flip the measurement outcomes by flipping successive $X$ checks.

\begin{align}\label{eq:upsilon_hom}
    \dim \im H^{2}(\Upsilon) &= \dim \im H^0(\iota)\dim\im H^2(fg) + \dim \im H^1(\iota)\dim\im H^1(fg) \\
    &= \begin{cases}
\dim\im H^1(fg), & \text{if case (1)}.\\
0, &  \text{if case (2)}.
\end{cases}
\end{align}
The (negative) contribution of $\dim \im H^{2}(\Upsilon)$ dictates that any new timelike logical faults that flip the  outcomes of $\bar{X}$ logicals that are being measured are equivalent to spacelike $\bar{Z}$ logical errors in the original code. 

\begin{remark}
    In general we may have $\dim H^1(\chi) > \dim\im H^1(fg)$, which simply means that there are extra timelike logical faults through the auxiliary system which flip $X$ checks but do not flip any $\bar{X}$ logical measurement outcomes: if there is a set of $X$ checks in $\chi^1$ which are mapped by $f^1g^1$ to a representative of $[0]$ in $H^1(C)$ then they do not contribute to any logical measurements. These are harmless, and the faults must have weight at least $d(C)$ to be undetected by detectors regardless.
\end{remark}

The reason why there is no contribution in case (2) is because, in that case, the code was initialised in the $\ket{+}$ state, and so all $\bar{X}$ logicals being measured are actually spacetime stabilisers.

For the final contributions, $\dim H^1(R\otimes C)$ comes from the original idling volume, and
\begin{align*}
    \dim \im H^{1}(\Upsilon) &= \dim \im H^0(\iota)\dim\im H^1(fg)
 = 0.
\end{align*}
Therefore the only new contributions are those of timelike errors, which extend from the beginning to the end of the ancilla system. As the surgery is performed for $d(C)$ rounds, i.e. $P^0 \rightarrow P^1 = \F_2^{d(C)} \rightarrow \F_2^{d(C)+1}$, the weight of these logical faults must be at least $d(C)$, even if they are cleaned by spacetime stabilisers or multiplied by other logical faults. Then, by~\cite[Thm.~1]{ide2025fault}, the existing nontrivial spacelike logicals in $H_1(R\otimes C)$ must have weight at least $d(C)$ even after cleaning into the ancilla system, so we are done.

\textbf{$\bar{X}$ faults:}
\begin{align*}
\dim H^2(\cone(\Upsilon)) &= \dim H^2(R\otimes C) - \dim \im H^2(\Upsilon) + \dim H^3(\Xi) - \dim \im H^{3}(\Upsilon) \\
&= \dim H^2(R\otimes C) - \dim \im H^2(\Upsilon) \\
&= \begin{cases}
 \dim H^2(R\otimes C) - \dim\im H^1(fg), & \text{if case (1)}.\\
\dim H^2(R\otimes C), &  \text{if case (2)}.
\end{cases} 
\end{align*}
where the last equality is due to Eq.~\ref{eq:upsilon_hom}.

Therefore the only equivalence classes of $\bar{X}$ logical faults in the spacetime volume $\cone(\Upsilon)$ are those present in the original idling volume, minus those which have been converted into stabilisers and thus measured out by the surgery procedure. In case (2) there are no spacelike $\bar{X}$ logicals in the original idling volume, so there is no negative contribution.

As $\chi$ has relative expansion $\beta_\tau = 1$ at degree 1 with $\tau = d(C)$, and $f^1g^1$ has column weight at most 1, no unmeasured spacelike $\bar{X}$ logicals can have their weight reduced below $d(C)$ by cleaning, following an identical argument to Lem.~\ref{lem:distance_preservation}, so we are done. \footnote{The only difference between this argument and Lem.~\ref{lem:distance_preservation} is that there, the fact that no new logicals were introduced came from injectivity of $H^2(fg)$; here it comes from the detectors $\Xi^4$.}

\endproof

\section{An improved linear PCPP}\label{app:improved_PCPP}

In this appendix, we outline a slightly better linear PCPP. Namely, we give a proof sketch for the following:

\begin{restatable}{lemma}{betterPCPP}\label{lem:better_PCPP}
    Let $\varphi : \F_2^m \rightarrow \F_2^p$ be a linear circuit of size $n$ over $\F_2$.
    Then there exists a linear PCPP verifier $V$ which takes as input $(\F_2, \varphi)$, with the following parameters:
    \begin{itemize}
        \item $r(n) = \log n + \mathcal{O}(\log \log n)$,
        \item $\ell(n) = n(\log n)^{\mathcal{O}(1)}$,
        \item $d(n) = (\log n)^{\mathcal{O}(1)}$,
        \item $\rho(n) = \Omega(\frac{1}{\log n})$.
    \end{itemize}
\end{restatable}
We believe that this lemma is already known, and is \textit{almost} stated verbatim in~\cite[Lem.~8.1]{bensasson2009sound}, but with some minor variation.
This lemma has the immediate consequence that our Lem.~\ref{lem:sparse_relative_expander} can be improved, yielding:

\begin{restatable}{lemma}{bettersparseexpander}
\label{lem:better_sparse_expander}
Let $A = \F_2^\mu \rightarrow \F_2^\nu$ be a length-1 cochain complex over $\F_2$ such that $\delta^1_A: \F_2^\mu \rightarrow \F_2^\nu$ is LDPC and $\nu \in \mathcal{O}(\mu)$.

    Then there exists a length-1 cochain complex $\chi = \chi^1 \rightarrow \chi^2$ over $\F_2$, along with cochain map $g: \chi \rightarrow A$, such that:
    \begin{enumerate}
        \item $g$ has row and column weights at most 1,
        \item $\dim \chi^1 + \dim \chi^2 = \mu(\log \mu)^{\mathcal{O}(1)}$,
        \item $\delta^1_\chi : \chi^1 \rightarrow \chi^2$ is LDPC,
        \item $g^1(\ker \delta^1_\chi) = \ker \delta^1_A $,
        \item $|\delta^1_\chi x| \geq \min_{z \in \ker\delta^1_\chi} |g^1(x-z)|$, $\forall x \in \chi^1$.
    \end{enumerate}
    Moreover, $\chi$ can be constructed in time polynomial in $\mu$.
\end{restatable}

First, we prove Lem.~\ref{lem:better_sparse_expander}, contingent on Lem.~\ref{lem:better_PCPP}.

\proof
We retrace the steps of Sec.~\ref{sec:lift}, but without having to do field extension or restriction. 
\begin{itemize}
    \item First, convert $\delta^1_A$ into an $\F_2$-linear circuit $\mathcal{A}$.
    \item Then, from the verifier $V$ of Lem.~\ref{lem:better_PCPP}, use Lem.~\ref{lem:matrix_realisation} to acquire a matrix $\hat{D}$, which realises the tester circuits and has sufficient relative expansion, following the proceeding arguments to boost the relative expansion to at least 1. $\hat{D}$ has $\mu(\log \mu)^{\mathcal{O}(1)}$ rows and columns; its columns may be dense, however.
    \item Then reduce the column weight by Lem.~\ref{lem:degree_reduction}, applied directly to $\hat{D}$ (skipping the field restriction) to acquire the sparsified matrix $D$. In our case, the size of $D$ is $\mathcal{O}\bigl(\nnz(\hat{D})+\mu + \dim \hat{P} \bigr) = \mu(\log \mu)^{\mathcal{O}(1)}$.
    \item Lastly, insert $D$ into the cochain complex in Sec.~\ref{sec:the_cone}.
\end{itemize}
\endproof

We now give a proof sketch of Lem.~\ref{lem:better_PCPP}.

\proofsketch
First, given $\varphi$, we will construct a Boolean PCPP (as opposed to a linear one). Ref.~\cite[Thm.~2.10]{bensasson2008short} gives a Boolean PCPP with nearly the same parameters as Lem.~\ref{lem:better_PCPP}. Explicitly, its parameters are,

\begin{itemize}
    \item $r_0 = \log n + \mathcal{O}(\log \log n)$,
    \item $\ell_0 = n (\log n)^{\mathcal{O}(1)}$,
    \item $q_0 = (\log n)^{\mathcal{O}(1)}$,
    \item $d_0 = n^{\mathcal{O}(1)}$,
    \item $\rho_0 = \Omega(\frac{1}{\log n})$.
\end{itemize}
The distinctions are that: (a) the tester circuits are Boolean circuits, not linear ones, and (b) the decision complexity is polynomial, not constant. Note that $q_0$ is the query complexity, i.e. the number of input bits to each tester circuit.

We can follow the observations of Ref.~\cite[Sec.~8]{bensasson2009sound}, which show that when the property being tested is $\F_2$-linear, the proof can be chosen as a linear function of the tested word: that is, there is an $\F_2$-linear map $f:\F_2^m \rightarrow \F_2^{\ell_0}$, where $m$ is the implicit input length, such that if $x \in \ker \varphi$ then $x \circ f(x)$ is always accepted by the tester. We claim that one can then replace each predicate in each Boolean tester circuit with a linear one, while preserving the completeness and soundness of the verifier. Some words in $\F_2^m\oplus \F_2^{\ell_0}$ which would have been accepted by the Boolean tester circuits are no longer accepted, but $x \circ f(x)$ is always accepted and so the completeness is preserved. No new words are accepted, and so the soundness is preserved. This linear PCPP now has decision complexity $d_1  = q_0^{\mathcal{O}(1)} = (\log n)^{\mathcal{O}(1)}$, by taking the polynomially large linear circuits, converting them into dense matrices which perform the same linear map up to row operations, preserving the kernel, then turning them back into linear circuits in the same manner as in Sec.~\ref{sec:circuit_construction}: large XORs are replaced by binary trees, and the same for fan-outs.
\endproofsketch

\end{document}